\documentclass[pdflatex,sn-mathphys-num]{sn-jnl}

\usepackage{graphicx}
\usepackage{amsmath,amssymb,amsfonts}
\usepackage{amsthm}
\usepackage{mathtools}
\usepackage{mathrsfs}
\usepackage[title]{appendix}
\usepackage{placeins}
\usepackage{xurl}
\usepackage{enumitem}
\usepackage{booktabs}
\usepackage{tabularx}
\usepackage{xcolor}
\usepackage{textcomp}
\usepackage{manyfoot}
\newcolumntype{Y}{>{\raggedright\arraybackslash}X}
\usepackage{etoolbox}
\hypersetup{hypertexnames=false}

\makeatletter
\patchcmd{\@maketitle}{Corresponding author(s). E-mail(s): }{E-mail: }{}{\ClassError{sn-jnl-patch}{e-mail line patch failed}{}}
\def\email#1{\global\advance\emailcnt by 1\relax%
\if@corauemail
   \g@addto@macro\corrauthemail{\setcounter{footnote}{0}\textcolor{blue}{#1}}%
\else
   \g@addto@macro\authemail{\setcounter{footnote}{0}\textcolor{blue}{#1}}%
\fi}
\makeatother

\newcommand{\vE}{\mathbf{E}}
\newcommand{\vB}{\mathbf{B}}
\newcommand{\vD}{\mathbf{D}}
\newcommand{\vH}{\mathbf{H}}
\newcommand{\vP}{\mathbf{P}}
\newcommand{\vM}{\mathbf{M}}
\newcommand{\vJ}{\mathbf{J}}
\newcommand{\vK}{\mathbf{K}}
\newcommand{\vA}{\mathbf{A}}
\newcommand{\vS}{\mathbf{S}}
\newcommand{\vg}{\mathbf{g}}
\newcommand{\vr}{\mathbf{r}}
\newcommand{\vp}{\mathbf{p}}
\newcommand{\zero}{\mathbf{0}}
\newcommand{\nhat}{\hat{\mathbf{n}}}
\newcommand{\khat}{\hat{\mathbf{k}}}
\newcommand{\rhat}{\hat{\mathbf{r}}}
\newcommand{\zhat}{\hat{\mathbf{z}}}
\newcommand{\eps}{\varepsilon_{0}}
\newcommand{\curl}{\nabla\times}
\newcommand{\dive}{\nabla\cdot}
\newcommand{\grad}{\nabla}
\newcommand{\pdt}[1]{\frac{\partial #1}{\partial t}}
\newcommand{\Ylm}{Y_{\ell m}}
\newcommand{\vI}{\mathbb{I}}
\newcommand{\Pperp}{\mathsf{P}_{\!\perp}}
\newcommand{\calR}{\mathcal{R}}
\newcommand{\calJ}{\mathcal{J}}

\theoremstyle{thmstyleone}
\newtheorem{theorem}{Theorem}
\newtheorem{lemma}{Lemma}
\newtheorem{proposition}{Proposition}
\newtheorem{corollary}{Corollary}

\theoremstyle{thmstyletwo}
\newtheorem{remark}{Remark}

\allowdisplaybreaks

\begin{document}
%Sorry to the arXiv team, there was an accidental intend here that I noticed which shifted the positioning of the affiliations with respect to my name.
\title[Scalar-Longitudinal Radiation in Extended Electrodynamics]
{Scalar-Longitudinal Radiation in Extended Electrodynamics with Multipole Theory and a Compensated Source Model}

\author*[1]{\fnm{Natan} \sur{Rentzber}}\email{nrentzbe@uccs.edu}

\affil*[1]{\orgdiv{Center for Magnetism and Magnetic Nanostructures},

\orgname{University of Colorado},
\orgaddress{\city{Colorado Springs}, \postcode{80918}, \country{USA}}}

\abstract{Extended electrodynamics (EED) leaves the Lorenz gauge condition unimposed and treats the scalar combination \(C=\nabla\!\cdot\!\mathbf A+c^{-2}\partial\Phi/\partial t\) as a dynamical field. For a conserved source with no scalar initial field, \(C=0\) and the theory reduces to classical electrodynamics. A source with a nonzero local continuity anomaly has no Maxwell solution, but EED remains well posed and can support a scalar-longitudinal sector. For each radiating frequency of a localized source, the far field separates into the usual transverse Maxwell channel and a scalar-longitudinal channel with a longitudinal electric field, no magnetic field of its own, and a co-propagating \(C\) field. Under the field-only energy balance used here, the time-averaged fluxes add without interference. The scalar channel depends only on \(\Lambda=\partial\rho/\partial t+\nabla\!\cdot\!\mathbf J\) and radiates when its moments at \(k=\omega/c\) are nonzero. An all-orders multipole formula is derived for this flux. Bound polarization and magnetization sources conserve charge identically and cannot excite the scalar channel. A compensated polarized carrier with a globally neutral anomalous surface layer isolates the channel and gives its dipole flux in closed form. The connection between the adopted flux and a physical stress-energy tensor remains unresolved. These results are conditional predictions of EED and do not imply a failure of charge conservation in classical electrodynamics.}


\maketitle

\clearpage
\setcounter{tocdepth}{2}
\tableofcontents
\clearpage

\section{Introduction}
\label{sec:intro}
In classical electrodynamics, meaning standard Maxwell theory throughout this paper, local charge conservation follows from the field equations. Taking the divergence of the Amp\`ere-Maxwell law and using Gauss's law gives $\partial\rho/\partial t+\dive\vJ=0$, so a prescribed source that violates local continuity has no Maxwell solution~\cite{Modanese2017}. Ohmura~\cite{Ohmura1956} instead treated $C=\dive\vA+c^{-2}\partial\Phi/\partial t$ as dynamical. Related formulations were developed by Aharonov and Bohm~\cite{AharonovBohm1963}, van~Vlaenderen and Waser~\cite{vanVlaenderen2001}, and Hively and Giakos~\cite{HivelyGiakos2012}. The framework is called extended electrodynamics (EED) here. Its potential equations remain solvable for prescribed localized sources, while the scalar field carries any continuity anomaly. For a conserved source with no incoming or homogeneous scalar field, $C=0$ and Maxwell theory is recovered.
\\

EED is nonstandard and is not treated as established physics. Charge conservation is tied to gauge invariance and is tightly constrained experimentally~\cite{Okun1989,Borexino2015,Majorana2024}. The aim is therefore conditional. The paper determines the radiation predicted by EED if a prescribed free source has a local anomaly. The main results are a source-independent radiation-zone decomposition, an exact spherical-multipole formula for the scalar channel, no-go results for bound and magnetization sources, and a reception identity for conserved test currents. Earlier work established longitudinal solutions or studied particular anomalous sources~\cite{ReedHively2020,MinottiModanese2021,MinottiModanese2021b,MinottiModanese2022}.
\\

Macroscopic electrodynamics provides a useful check. The standard bound pair $\rho_b=-\dive\vP$ and $\vJ_b=\partial\vP/\partial t+\curl\vM$ conserves charge identically, including the interface terms. Hence the auxiliary fields $\vD$ and $\vH$ add no independent radiative degree of freedom~\cite{Jackson,Griffiths}. A source whose total $\rho$ and $\vJ$ vanish as distributions has no retarded source-generated field in either Maxwell theory or EED.
\\

Section~\ref{sec:decomp} derives the two-channel decomposition. Section~\ref{sec:multipole} gives the scalar multipole formula and its long-wavelength limits. Section~\ref{sec:corollaries} develops the source and detector corollaries. Sections~\ref{sec:electret} and~\ref{sec:compare} construct the compensated source and compare it with Maxwell radiation.
\\

SI units and real time-domain fields are used. Interface components are labeled perpendicular ($\perp$) and parallel ($\parallel$) to a surface. For a continuous Fourier variable $\Omega$,
\begin{equation}
f(\vr,\Omega)=\int_{-\infty}^{\infty}f(\vr,t)e^{i\Omega t}\,dt,
\qquad
f(\vr,t)=\frac{1}{2\pi}\int_{-\infty}^{\infty}f(\vr,\Omega)e^{-i\Omega t}\,d\Omega.
\label{eq:fourier-convention}
\end{equation}
The element $d\Omega$ under frequency integrals refers to this transform variable, while under angular integrals the same symbol denotes the solid angle. A selected monochromatic component at $\omega>0$ is written
\begin{equation}
f(\vr,t)=\mathrm{Re}\!\left[f_{\omega}(\vr)e^{-i\omega t}\right],
\label{eq:phasor-convention}
\end{equation}
where $f_{\omega}$ is a finite phasor amplitude, not the full transform. For an eternal sinusoid,
\begin{equation}
f(\vr,\Omega)=\pi\!\left[f_{\omega}(\vr)\delta(\Omega-\omega)+f_{\omega}^{*}(\vr)\delta(\Omega+\omega)\right],
\end{equation}
and real fields satisfy $f(\vr,-\Omega)=f(\vr,\Omega)^{*}$. This convention is consistent with the frequency-domain treatment used in standard electrodynamics~\cite{Jackson}.
\\

Here ``on shell'' means evaluation at spatial wave vector $k\rhat$, with $k=\omega/c$. The spherical harmonics are the orthonormal complex $Y_{\ell m}$ with the Condon-Shortley phase,
$\oint Y_{\ell m}Y^{*}_{\ell' m'}\,d\Omega=\delta_{\ell\ell'}\delta_{mm'}$.
The transverse projector is $\Pperp(\rhat)=\vI-\rhat\otimes\rhat$. Appendices~\ref{app:poynting}--\ref{app:lagrangian} contain the main derivations.

\section{Extended Electrodynamics}
\label{sec:eed}

\subsection{Field Equations}

EED starts from the sourced wave equations below without imposing the Lorenz condition.
\begin{equation}
\nabla^{2}\Phi-\frac{1}{c^{2}}\frac{\partial^{2}\Phi}{\partial t^{2}}=-\frac{\rho}{\eps},\qquad
\nabla^{2}\vA-\frac{1}{c^{2}}\frac{\partial^{2}\vA}{\partial t^{2}}=-\mu_{0}\vJ.
\label{eq:potwave}
\end{equation}
The electric and magnetic fields keep their usual definitions, $\vE=-\grad\Phi-\partial\vA/\partial t$ and $\vB=\curl\vA$. Under these conventions, $C=\dive\vA+c^{-2}\partial\Phi/\partial t$ has the units of magnetic field, so $cC$ has the units of electric field. The homogeneous equations $\curl\vE=-\partial\vB/\partial t$ and $\dive\vB=0$ hold identically. Taking $\dive\vE$ and $\curl\vB$ and using \eqref{eq:potwave} cancels the $\partial^{2}\Phi/\partial t^{2}$ and $\partial^{2}\vA/\partial t^{2}$ terms against $C$. The extended Gauss and Amp\`ere laws follow.
\begin{equation}
\dive\vE=\frac{\rho}{\eps}-\pdt{C},\qquad
\curl\vB-\frac{1}{c^{2}}\pdt{\vE}=\mu_{0}\vJ+\grad C.
\label{eq:eedfields}
\end{equation}
These equations reduce to Maxwell's equations when $C=0$~\cite{Ohmura1956,vanVlaenderen2001,HivelyGiakos2012}.

\subsection{Scalar Wave Equation}

Applying $\nabla^{2}-c^{-2}\partial^{2}/\partial t^{2}$ to the definition of $C$ and using
\eqref{eq:potwave}, $1/(c^{2}\eps)=\mu_{0}$, gives
\begin{equation}
\nabla^{2}C-\frac{1}{c^{2}}\frac{\partial^{2}C}{\partial t^{2}}=-\mu_{0}\,\Lambda,\qquad \Lambda\equiv\pdt{\rho}+\dive\vJ.
\label{eq:scalarbox}
\end{equation}
Only the charge-conservation anomaly $\Lambda$ sources the scalar field. If $\Lambda=0$ and no incoming or homogeneous scalar field is prescribed, the retarded source-generated solution has $C\equiv0$ and Maxwell theory is recovered. Here zero scalar initial data means that no independent incoming or homogeneous scalar field is present. Homogeneous solutions of the scalar equation, including the scalar-longitudinal wave below, are mathematically allowed, but a conserved source does not generate them from zero initial data.

\subsection{Prescribed Sources and Well-Posedness}

The extended system is well posed for any sufficiently regular localized prescribed pair $(\rho,\vJ)$. This includes compactly supported smooth sources and compactly supported distributions for which the retarded convolution exists. Monochromatic idealizations use the outgoing limiting prescription. Equations~\eqref{eq:potwave} are four decoupled inhomogeneous wave equations with unique retarded solutions. Equations~\eqref{eq:eedfields} then follow identically, while \eqref{eq:scalarbox} shows that $C$ carries the full conservation anomaly. Maxwell theory instead requires $\Lambda=0$ because $\dive(\curl\vB)=0$ and Gauss's law impose it as an integrability condition. A nonconserving source is therefore not an unusual Maxwell radiator. Maxwell's equations have no solution for it~\cite{Modanese2017}. EED remains solvable for such prescribed sources. As in antenna and multipole theory, the sources are treated as formal distributions. Section~\ref{sec:decomp} gives their leading radiation-zone fields.

\subsection[Covariance of C and its Relation to Gauge Fixing]{Covariance of $C$ and its Relation to Gauge Fixing}

The potential wave equations \eqref{eq:potwave}, together with the usual definitions of $\vE$, $\vB$, and $C=\dive\vA+c^{-2}\partial\Phi/\partial t$, are the starting point. In Maxwell theory, $C$ is a gauge quantity. The Lorenz condition may set $C=0$ without changing observable fields. EED leaves that condition unimposed and treats $C$ as a source-driven field through \eqref{eq:scalarbox}. This postulate has a covariant form. Equations~\eqref{eq:potwave} are the Euler-Lagrange equations of
\begin{equation}
\mathcal{L}_{\mathrm{EED}}=-\frac{1}{4\mu_{0}}F_{\mu\nu}F^{\mu\nu}-\frac{1}{2\mu_{0}}\,C^{2}-J^{\mu}A_{\mu},
\qquad C=\partial_{\mu}A^{\mu}.
\label{eq:lagrangian}
\end{equation}
Here the metric signature is $(+,-,-,-)$, $A^{\mu}=(\Phi/c,\vA)$, and $J^{\mu}=(c\rho,\vJ)$. With these definitions, $\partial_{\mu}\partial^{\mu}=c^{-2}\partial_{t}^{2}-\nabla^{2}$ is the negative of the operator in \eqref{eq:potwave}. Varying the action gives $\partial_{\mu}F^{\mu\nu}+\partial^{\nu}C=\mu_{0}J^{\nu}$, or equivalently $\partial_{\mu}\partial^{\mu}A^{\nu}=\mu_{0}J^{\nu}$. This is Eq.~\eqref{eq:potwave}. Taking the four-divergence and using $\partial_{\nu}\partial_{\mu}F^{\mu\nu}\equiv0$ gives \eqref{eq:scalarbox}. Since $C=\partial_{\mu}A^{\mu}$ is a Lorentz scalar, the extended theory is Lorentz covariant. The transverse-longitudinal split below is made in the asymptotic frame of the localized source, as in standard multipole theory.
\\

Up to a total four-divergence, the field part of \eqref{eq:lagrangian} is the covariant Fermi, or Feynman-gauge, Lagrangian~\cite{Fermi1932}, and a related treatment is given in Ref.~\cite{Stueckelberg1938}. The $C^{2}$ term can also be obtained by eliminating a Nakanishi-Lautrup auxiliary field~\cite{Nakanishi1966,Lautrup1967}. EED differs in interpretation and treats $C=\partial_{\mu}A^{\mu}$ as physical rather than as a gauge-fixing quantity. For a conserved source with retarded boundary conditions and no scalar initial field, $C=0$ and no additional observable field remains. Residual transformations satisfying $\partial_{\nu}\partial^{\nu}\chi=0$ are discussed in Appendix~\ref{app:lagrangian}. Some EED formulations define the scalar with the opposite sign~\cite{HivelyGiakos2012}. That convention reverses terms linear in $C$ but leaves quadratic fluxes and powers unchanged. The same quadratic divergence term has also been treated as physical in cosmological models~\cite{JimenezMaroto2009,JimenezMaroto2011}.
\\

When $\Lambda\neq0$, setting $C=0$ is not a gauge choice for the same source-coupled problem. Under $A_{\mu}\to A_{\mu}+\partial_{\mu}\chi$, the interaction action changes by $c^{-1}\!\int\chi\Lambda\,d^{4}x$. Arbitrary-gauge invariance of the coupling and local charge conservation therefore fail together. EED then assumes that the generated $C$ field is physical and assigns it the field-energy balance below. The radiation amplitudes derived later follow from the field equations, while statements about power and energy also depend on that adopted balance.

\subsection{Poynting Balance}

The extended Poynting theorem reads $\partial u/\partial t+\dive\vS
=-\vE\cdot\vJ+c^{2}\rho C$, with
\begin{equation}
u=\frac{\eps}{2}|\vE|^{2}+\frac{1}{2\mu_{0}}|\vB|^{2}+\frac{1}{2\mu_{0}}C^{2},\qquad
\vS=\frac{1}{\mu_{0}}\big(\vE\times\vB\big)+\frac{1}{\mu_{0}}\,C\,\vE.
\label{eq:uS}
\end{equation}
The balance follows algebraically from the field equations. Whether its flux is the flux of a physical EED stress-energy tensor remains unresolved. Accordingly, ``power,'' ``energy,'' and ``flux'' below refer to \eqref{eq:uS}. The far-zone fields and multipole expansion of $C$ do not depend on this choice. The scalar terms are $C^{2}/(2\mu_{0})$ and $C\vE/\mu_{0}$. Section~\ref{sec:discussion} compares this functional with canonical and tensor treatments.

\subsection{Free Scalar-Longitudinal Wave}

In vacuum, Eq.~\eqref{eq:eedfields} admits a longitudinal solution in addition to the two transverse electromagnetic waves. Let $\vE=E_{0}\cos(kz-\omega t)\zhat$, $\vB=\zero$, and $C=C_{0}\cos(kz-\omega t)$, with
\begin{equation}
\omega^{2}=c^{2}k^{2},\qquad C_{0}=E_{0}/c.
\label{eq:slw}
\end{equation}
\\
This scalar-longitudinal wave (SLW) has a longitudinal electric field, no magnetic field, and a co-propagating scalar with a fixed amplitude ratio. One potential representation has phasors $\Phi_{\omega}=0$ and $\vA_{\omega}=(E_{0}/i\omega)\,e^{ikz}\,\zhat$, giving $\vB_{\omega}=\curl\vA_{\omega}=\zero$ and $C_{\omega}=\partial_{z}A_{\omega,z}=(E_{0}/c)\,e^{ikz}$, consistent with \eqref{eq:slw}. Figure~\ref{fig:structure} compares this branch with an ordinary transverse Maxwell wave. Equation~\eqref{eq:uS}, together with $\eps c^{2}=1/\mu_{0}$, gives $u=\eps E_{0}^{2}\cos^{2}(kz-\omega t)$ and $\vS=cu\,\zhat$, with $\vg=u\zhat/c$. Under the adopted balance, the wave transports energy even though $\vB=\zero$. By contrast, a vacuum configuration with $\vE=\vB=\zero$ forces $\partial C/\partial t=0$ and $\grad C=0$ through $\dive\vE=-\partial C/\partial t$ and $c^{-2}\partial\vE/\partial t=-\grad C$. Such a field is spatially uniform, time independent, and carries no flux. Unless confined or subtracted, it is not part of the localized radiation problem.
\\

\begin{remark}[Uniqueness of the Radiative Scalar Branch]
Within the field-energy accounting used here, the SLW is the only propagating scalar branch generated by localized sources from zero scalar initial data. Some papers give separate status to a ``scalar wave''~\cite{HivelyGiakos2012,ReedHively2020}. In the source-driven setting here, the $\vE=\vB=\zero$ argument restricts it to a static background with no flux. Pure potential waves with $\vE=\vB=\zero$ and $C=0$, sometimes called ``gauge waves,'' have also been discussed~\cite{MinottiModanese2023,MinottiModanese2024}. They carry no energy under \eqref{eq:uS}, although anomalous probes could in principle couple directly to the potentials. Such waves lie outside the field-level energy accounting used here.
\end{remark}

\begin{figure}[t]
\centering
\includegraphics[width=0.98\textwidth]{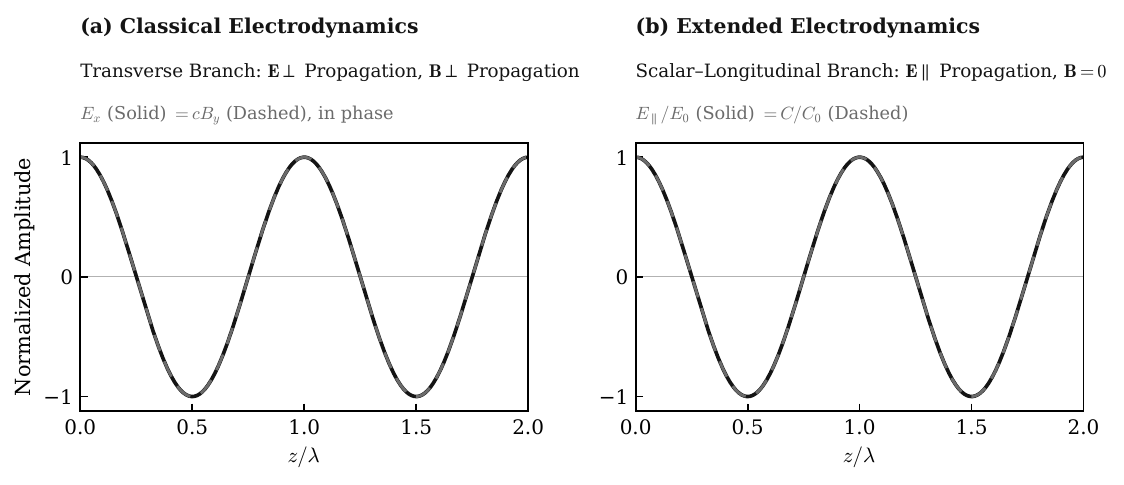}
\caption{Radiative field structure in classical electrodynamics and EED.
\textbf{(a)} An ordinary Maxwell plane wave is transverse, with $\vE\perp\khat$ and $\vB\perp\khat$. The fields $\vE$ and $c\vB$ are in phase.
\textbf{(b)} The EED scalar-longitudinal branch has $\vE\parallel\khat$, $\vB=\zero$, and a co-propagating scalar field $C$ in phase with the longitudinal electric field. Its amplitudes satisfy $C_{0}=E_{0}/c$. Both branches travel at $c$.}
\label{fig:structure}
\end{figure}

\section{Two-Channel Decomposition Theorem}
\label{sec:decomp}
Let the sources lie in a bounded region and consider one component with $\omega>0$, time dependence $e^{-i\omega t}$, $k=\omega/c$, and observation direction $\rhat$. All quantities in this section with subscript $\omega$ are phasor amplitudes. Define
\begin{equation}
\calR_{\omega}(\rhat)=\int \rho_{\omega}(\vr')e^{-ik\rhat\cdot\vr'}\,d^{3}r',
\qquad
\calJ_{\omega}(\rhat)=\int \vJ_{\omega}(\vr')e^{-ik\rhat\cdot\vr'}\,d^{3}r',
\label{eq:on-shell-rhoJ}
\end{equation}
and
\begin{equation}
\Lambda_{\mathrm{on},\omega}(\rhat)=\int \Lambda_{\omega}(\vr')e^{-ik\rhat\cdot\vr'}\,d^{3}r'
=-i\omega\calR_{\omega}(\rhat)+ik\,\rhat\cdot\calJ_{\omega}(\rhat).
\label{eq:on-shell-Lambda}
\end{equation}
The last equality follows by integrating the divergence term by parts. The surface term vanishes for localized sources. With $\Pperp(\rhat)=\vI-\rhat\otimes\rhat$, the retarded potentials give (Appendix~\ref{app:farfield})
\begin{align}
C_{\omega}(\vr)&=\frac{\mu_{0}e^{ikr}}{4\pi r}\,\Lambda_{\mathrm{on},\omega}(\rhat)+O(r^{-2}),
\label{eq:farC}\\
\vE_{T,\omega}(\vr)&=\frac{i\omega\mu_{0}e^{ikr}}{4\pi r}\,\Pperp(\rhat)\calJ_{\omega}(\rhat)+O(r^{-2}),
\label{eq:farET}\\
\vB_{T,\omega}(\vr)&=\frac{1}{c}\,\rhat\times\vE_{T,\omega}(\vr)+O(r^{-2}),
\label{eq:farBT}\\
\vE_{L,\omega}(\vr)&=c\,C_{\omega}(\vr)\,\rhat+O(r^{-2}).
\label{eq:farEL}
\end{align}
Here ``transverse current'' means the on-shell amplitude $\Pperp(\rhat)\calJ_{\omega}(\rhat)$, not a global Helmholtz component of $\vJ$.

\begin{theorem}[Orthogonal Channel Decomposition]
\label{thm:decomp}
For each monochromatic component, the radiation-zone phasors separate as
\begin{equation}
\vE_{\mathrm{rad},\omega}=\vE_{T,\omega}+\vE_{L,\omega},
\qquad
\vE_{L,\omega}=cC_{\omega}\rhat,
\qquad
\vB_{\mathrm{rad},\omega}=\vB_{T,\omega},
\label{eq:fielddecomp}
\end{equation}
with $\vB_{L,\omega}=\zero$. The time-averaged power is
\begin{equation}
\langle P\rangle=\langle P_T\rangle+\langle P_L\rangle,
\label{eq:power-split-a}
\end{equation}
where
\begin{align}
\langle P_T\rangle&=\frac{1}{2\mu_{0}}\oint
\mathrm{Re}\!\left[(\vE_{T,\omega}\times\vB_{T,\omega}^{*})\cdot\rhat\right]r^2\,d\Omega,
\\
\langle P_L\rangle&=\frac{c}{2\mu_{0}}\oint |C_{\omega}|^2r^2\,d\Omega.
\label{eq:power-split}
\end{align}
\end{theorem}
\begin{proof}
At order $1/r$, $\vE_{T,\omega}\cdot\rhat=\vB_{T,\omega}\cdot\rhat=0$ and $\vE_{L,\omega}\parallel\rhat$. For the corresponding real fields, insert $\vE=\vE_T+cC\rhat$ and $\vB=\vB_T$ into \eqref{eq:uS}. Then
\[
[(\vE_T+cC\rhat)\times\vB_T]\cdot\rhat=(\vE_T\times\vB_T)\cdot\rhat,
\qquad
C(\vE_T+cC\rhat)\cdot\rhat=cC^2.
\]
Thus the radial flux has no transverse-longitudinal cross term. Terms involving an $r^{-2}$ field vanish after integration over the far sphere. Cycle averaging gives \eqref{eq:power-split-a}--\eqref{eq:power-split}. The transverse channel depends on $\Pperp(\rhat)\calJ_{\omega}(\rhat)$, whereas the scalar channel depends only on $\Lambda_{\mathrm{on},\omega}(\rhat)$.
\\
\end{proof}

The leading phasor Poynting vector is
\begin{equation}
\langle\vS_{\mathrm{rad}}\rangle=
\left[
\frac{|\vE_{T,\omega}|^2}{2c\mu_{0}}
+\frac{c|C_{\omega}|^2}{2\mu_{0}}
\right]\rhat+O(r^{-3}).
\label{eq:phasor-flux}
\end{equation}
The mixed tangential terms cancel as well. In phasor form, $\vB_{T,\omega}^{*}=c^{-1}\rhat\times\vE_{T,\omega}^{*}$ gives
$cC_{\omega}\,\rhat\times\vB_{T,\omega}^{*}=C_{\omega}\big[\rhat(\rhat\cdot\vE_{T,\omega}^{*})-\vE_{T,\omega}^{*}\big]=-C_{\omega}\vE_{T,\omega}^{*}$, so $cC_{\omega}\,\rhat\times\vB_{T,\omega}^{*}+C_{\omega}\vE_{T,\omega}^{*}=\zero$. The same cancellation holds for the corresponding real fields, as used in the proof.
\\

\begin{remark}[Scope of the Theorem]
The theorem assumes localized sources and the outgoing retarded solution. For a finite-duration square-integrable signal, define $k_{\Omega}=\Omega/c$ and the full-transform moments
\begin{equation}
q^{\mathrm{FT}}_{\ell m}(\Omega)=\int j_{\ell}(k_{\Omega}r')Y_{\ell m}^{*}(\rhat')\Lambda(\vr',\Omega)\,d^{3}r'.
\end{equation}
Excluding eternal monochromatic delta distributions, Parseval's theorem gives
\begin{equation}
W_L=\frac{c\mu_{0}}{\pi}\int_{0}^{\infty}\sum_{\ell m}|q^{\mathrm{FT}}_{\ell m}(\Omega)|^2\,d\Omega.
\label{eq:WL}
\end{equation}
Cross-frequency terms vanish after time integration, while the channel cross term vanishes pointwise.
\\
\end{remark}

\begin{corollary}[Uniqueness of the New Channel]
\label{cor:unique}
Within the source-generated field sector considered here, EED adds one flux-carrying scalar-longitudinal branch to the two transverse Maxwell polarizations. Pure-potential configurations with no field energy are excluded from this field-strength radiation content.
\end{corollary}

\section{Multipole Theory of the Scalar Channel}
\label{sec:multipole}
For a localized harmonic anomaly, \eqref{eq:scalarbox} becomes
$(\nabla^2+k^2)C_{\omega}=-\mu_{0}\Lambda_{\omega}$.
The outgoing Green function and its spherical expansion~\cite{Jackson} give, outside the source,
\begin{equation}
\begin{aligned}
C_{\omega}(\vr)
&=i\mu_{0}k\sum_{\ell=0}^{\infty}\sum_{m=-\ell}^{\ell}
q_{\ell m,\omega}\,h_{\ell}^{(1)}(kr)Y_{\ell m}(\rhat),\\
q_{\ell m,\omega}
&=\int j_{\ell}(kr')Y_{\ell m}^{*}(\rhat')\Lambda_{\omega}(\vr')\,d^3r'.
\end{aligned}
\label{eq:Cmultipole}
\end{equation}
Since $h_{\ell}^{(1)}(kr)\to(-i)^{\ell+1}e^{ikr}/(kr)$,
\begin{equation}
C_{\omega}(\vr)\to\frac{\mu_{0}e^{ikr}}{r}
\sum_{\ell m}(-i)^{\ell}q_{\ell m,\omega}Y_{\ell m}(\rhat).
\label{eq:Cmultipole-far}
\end{equation}

\begin{proposition}[Scalar Radiated Power at All Multipole Orders]
\label{prop:power}
The exact time-averaged scalar-longitudinal power at the selected frequency is
\begin{equation}
\boxed{\langle P_L\rangle=\frac{c\mu_{0}}{2}
\sum_{\ell=0}^{\infty}\sum_{m=-\ell}^{\ell}|q_{\ell m,\omega}|^2}.
\label{eq:Pmaster}
\end{equation}
For a source of radius $R$ with $kR\ll1$,
\begin{equation}
\langle P_{\ell}\rangle\simeq
\frac{c\mu_{0}}{2}\frac{k^{2\ell}}{[(2\ell+1)!!]^2}
\sum_{m=-\ell}^{\ell}|\mathcal M_{\ell m,\omega}|^2,
\qquad
\mathcal M_{\ell m,\omega}=\int r'^{\ell}Y_{\ell m}^{*}(\rhat')\Lambda_{\omega}(\vr')\,d^3r'.
\label{eq:Pell}
\end{equation}
\end{proposition}
\begin{proof}
Write \eqref{eq:Cmultipole-far} as $C_{\omega}=e^{ikr}f(\rhat)/r$. Equation~\eqref{eq:power-split} and spherical-harmonic orthonormality give
$\oint|f|^2d\Omega=\mu_{0}^2\sum_{\ell m}|q_{\ell m,\omega}|^2$, which yields \eqref{eq:Pmaster}. The leading term
$j_{\ell}(kr')\simeq(kr')^{\ell}/(2\ell+1)!!$ gives \eqref{eq:Pell}.
\\
\end{proof}

\begin{corollary}[Long-Wavelength Multipole Hierarchy]
\label{cor:hierarchy}
For comparable reduced moment norms, successive multipoles satisfy
$\langle P_{\ell+1}\rangle/\langle P_{\ell}\rangle\sim(kR)^2$.
More explicitly, within the leading long-wavelength approximation,
\begin{equation}
\frac{\langle P_{\ell+1}\rangle}{\langle P_{\ell}\rangle}
\simeq\frac{k^2}{(2\ell+3)^2}
\frac{\sum_m|\mathcal M_{\ell+1,m,\omega}|^2}
{\sum_m|\mathcal M_{\ell m,\omega}|^2}.
\label{eq:multipole-ratio}
\end{equation}
Here ``comparable'' means
$\sum_m|\mathcal M_{\ell+1,m,\omega}|^2\sim R^2\sum_m|\mathcal M_{\ell m,\omega}|^2$.
Symmetry or cancellation can change the ordering.
\end{corollary}

\subsection[Monopole (l = 0)]{Monopole ($\ell=0$)}
In the long-wavelength limit,
$\mathcal M_{00,\omega}=q_{\Lambda,\omega}/\sqrt{4\pi}$, where
$q_{\Lambda,\omega}=\int\Lambda_{\omega}\,d^3r'$ is the phasor amplitude of $dQ/dt$.
For $Q(t)=Q_{0}\sin\omega t$, $dQ/dt=\lambda(t)=\lambda_{0}\cos\omega t$ with $\lambda_{0}=Q_{0}\omega$. Equation~\eqref{eq:Pell} gives
\begin{equation}
\langle P\rangle_{\mathrm{monopole}}\simeq
\frac{\mu_{0}c\lambda_0^2}{8\pi}
=\frac{\mu_{0}cQ_0^2\omega^2}{8\pi}.
\label{eq:Pmono}
\end{equation}
\\
This mode requires oscillation of the total charge and is therefore excluded by Maxwell theory.

\subsection[Dipole (l = 1)]{Dipole ($\ell=1$)}
Using
$\sum_m|\mathcal M_{1m,\omega}|^2=(3/4\pi)|\vp_{\Lambda,\omega}|^2$ and
$\vp_{\Lambda,\omega}=\int\vr'\Lambda_{\omega}\,d^3r'$, the leading dipole power is
\begin{equation}
\langle P\rangle_{\mathrm{dipole}}\simeq
\frac{c\mu_{0}k^2}{24\pi}|\vp_{\Lambda,\omega}|^2
=\frac{\mu_{0}}{12\pi c}
\left\langle\left|\frac{d\vp_{\Lambda}}{dt}\right|^2\right\rangle.
\label{eq:Pdip}
\end{equation}
The exact time-domain anomaly dipole obeys
\begin{equation}
\vp_{\Lambda}=\int\vr\Lambda\,d^3r
=\frac{d}{dt}\int\vr\rho\,d^3r-\int\vJ\,d^3r.
\label{eq:panomaly}
\end{equation}
Thus a conserved source satisfies $d\vp/dt=\int\vJ\,d^3r$, while a currentless anomalous source has $\vp_{\Lambda}=d\vp/dt$. For a general complex dipole amplitude,
\begin{equation}
\left\langle\frac{dP_L}{d\Omega}\right\rangle_{\mathrm{dipole}}
\simeq\frac{c\mu_{0}k^2}{32\pi^2}
|\rhat\cdot\vp_{\Lambda,\omega}|^2.
\label{eq:dPdip}
\end{equation}
For $\vp_{\Lambda,\omega}=p_{\Lambda,0}\zhat$, this becomes
$\frac{c\mu_{0}k^{2}}{32\pi^{2}}|p_{\Lambda,0}|^{2}\cos^{2}\theta$.
The pattern is strongest along the dipole axis, unlike the $\sin^2\theta$ pattern of an ordinary electric dipole. Integration over solid angle returns \eqref{eq:Pdip}.

\section{Corollaries for Bound Sources, Asymmetry, and Detectors}
\label{sec:corollaries}

Write the total sources as $\rho=\rho_{f}+\rho_{b}$ and $\vJ=\vJ_{f}+\vJ_{b}$. The standard bound sources are
$\rho_{b}=-\dive\vP$ and $\vJ_{b}=\partial\vP/\partial t+\curl\vM$.

\begin{lemma}[The Bound Charge-Current Pair is Identically Conserved]
\label{lem:bound}
$\Lambda_{b}\equiv\partial\rho_{b}/\partial t+\dive\vJ_{b}=0$ for any smooth $\vP,\vM$. The identity also holds distributionally for piecewise-smooth media with fixed interfaces.
\end{lemma}
\begin{proof}
Direct substitution gives $\Lambda_{b}=-\dive(\partial\vP/\partial t)+\dive(\partial\vP/\partial t)+\dive(\curl\vM)=0$. The first two terms cancel because mixed partial derivatives commute, and $\dive(\curl\vM)\equiv0$. For piecewise-smooth media with fixed interfaces, the same identity holds distributionally when the standard bound surface densities are included. Moving interfaces add convective surface terms and are outside the present scope. The surface term $\partial_{t}\sigma_{b}$, with $\sigma_{b}=\nhat\cdot(\vP_{1}-\vP_{2})$ and $\nhat$ directed from region 1 to region 2, is canceled by the normal jump of the polarization current $\partial\vP/\partial t$. At a medium-vacuum interface, $\vP_{2}=\zero$. A discontinuous polarization, including the electret sphere below, therefore still satisfies the identity.
\end{proof}

\begin{corollary}[No-Go Result for Bound Sources]
\label{cor:nogo}
Bound charge and current do not source the scalar channel. The flux $\langle P_{L}\rangle$ depends only on $\Lambda=\Lambda_{f}=\partial_{t}\rho_{f}+\dive\vJ_{f}$. If the free sources obey local continuity, then $\Lambda_{f}=0$. This includes a body with no free sources. With zero scalar initial data, such a body has $C\equiv0$ and follows Maxwell's equations. The fields $\vD$ and $\vH$ therefore gain no additional radiative degrees of freedom.
\end{corollary}
\medskip

\noindent A time-dependent polarization current $\partial\vP/\partial t$ contributes $\dive(\partial\vP/\partial t)$ to $\vJ_{b}$, but $\partial\rho_{b}/\partial t$ cancels it exactly. Only an anomalous free-source sector can drive the scalar wave. That sector may contain anomalous free charge, anomalous current, or both.
\\

The macroscopic equations give the same result directly. With $\vD=\eps\vE+\vP$ and $\vH=\vB/\mu_{0}-\vM$, Eq.~\eqref{eq:eedfields} becomes
\begin{equation}
\dive\vD=\rho_{f}-\eps\,\frac{\partial C}{\partial t},\qquad
\curl\vH-\frac{\partial\vD}{\partial t}=\vJ_{f}+\frac{1}{\mu_{0}}\grad C.
\label{eq:macro}
\end{equation}
Taking $\partial_{t}(\dive\vD)+\dive(\curl\vH-\partial_{t}\vD)$ gives $(\nabla^{2}-c^{-2}\partial_{t}^{2})C=-\mu_{0}(\partial_{t}\rho_{f}+\dive\vJ_{f})$. The bound terms cancel, so only the free sector sources the scalar field. Here ``auxiliary'' means that $\vD$ and $\vH$ introduce no new radiative degrees of freedom. The constitutive fields $\vP$ and $\vM$ may still have their own material dynamics.
\\

\begin{corollary}[Electric-Magnetic Asymmetry]
\label{cor:asymmetry}
A compensated magnetized body with no local continuity anomaly cannot radiate in the scalar channel. The anomaly in \eqref{eq:scalarbox} contains the charge rate $\partial\rho/\partial t$ and the longitudinal part of the current. A body with $\rho=0$ and only transverse curl currents has $\Lambda\equiv0$, so $\langle P_{L}\rangle=0$. A magnetic counterpart would require magnetic charge or a dual scalar sector. Neither is present in the EED model used here.
\end{corollary}
\medskip

This asymmetry follows from the equations, not from a particular drive. EED adds one scalar to the electric-longitudinal sector, and the charge anomaly sources it. Magnetization currents are solenoidal and provide no anomaly. The new channel therefore lies in a charge-nonconserving free-source sector rather than in bound matter or magnetization. This leads to the source construction in the next section.

\begin{corollary}[Vanishing Adopted-Current Overlap for a Conserved Test Current]
\label{cor:detector}
Let a pure incident scalar-longitudinal phasor $(\vE_{L,\mathrm{inc},\omega},C_{\mathrm{inc},\omega})$ solve the source-free equations before a localized detector is added. Neglect detector self-fields and scattering. For a conserved detector pair $(\rho_{d,\omega},\vJ_{d,\omega})$ with $\Lambda_{d,\omega}=0$,
\begin{equation}
\langle P_d\rangle=\frac12\mathrm{Re}\!\int
\left(\vE_{L,\mathrm{inc},\omega}\cdot\vJ_{d,\omega}^{*}
-c^2C_{\mathrm{inc},\omega}\rho_{d,\omega}^{*}\right)d^3r=0.
\label{eq:detectorzero}
\end{equation}
For a general detector,
\begin{equation}
\langle P_d\rangle=\frac12\mathrm{Re}\!\left[
\frac{ic^2}{\omega}\int C_{\mathrm{inc},\omega}\Lambda_{d,\omega}^{*}\,d^3r
\right].
\label{eq:detector-general}
\end{equation}
\end{corollary}
\begin{proof}
In a source-free region with $\vB_{\omega}=\zero$, the extended Amp\`ere law gives
$\vE_{L,\omega}=(c^2/i\omega)\grad C_{\omega}$. Integration by parts yields
\[
\int\vE_{L,\omega}\cdot\vJ_{d,\omega}^{*}d^3r
=-\frac{c^2}{i\omega}\int C_{\omega}(\dive\vJ_{d,\omega})^{*}d^3r.
\]
Using $\dive\vJ_{d,\omega}=\Lambda_{d,\omega}+i\omega\rho_{d,\omega}$ gives \eqref{eq:detector-general}. Setting $\Lambda_{d,\omega}=0$ gives \eqref{eq:detectorzero}.
\end{proof}

This is a test-current result, not a complete receiver theory. A conserved prescribed current extracts no net cycle-averaged energy under the adopted balance, although an assumed Lorentz force can still produce instantaneous mechanical response. Active, nonlinear, parametric, and back-reacting detectors require separate models. Proposed anomalous receivers and alternative matter couplings are discussed in Refs.~\cite{HivelyGiakos2012,ReedHively2020,MinottiModanese2024,MinottiModanese2026,MinottiModanese2026b}. For ordinary charge-conserving media, Ref.~\cite{MinottiModanese2023b} likewise finds that the scalar field does not propagate or dissipate energy, although a pure potential wave may occur.

\section{A Compensated Polarized Source With a Globally Neutral Anomalous Shell}
\label{sec:electret}

A compensated carrier can remove the transverse Maxwell background, while a prescribed anomalous surface layer excites only the scalar channel. The anomalous sector is not a Maxwell-admissible source and is used only as an EED diagnostic. Figure~\ref{fig:electret} summarizes the construction.

\subsection{Maxwell-Silent Carrier}

Consider a sphere of radius $R$ with uniform polarization $\vP(t)=P_{0}\cos(\omega t)\zhat$ and no magnetization. This is an externally driven harmonic idealization of an electret, whose physical polarization would normally be nearly static. It is a prescribed kinematic source rather than a self-consistent material model. The bound sources are $\rho_{b}=0$ in the bulk, $\sigma_{b}=P_{0}\cos\omega t\cos\theta$ on the surface, and the bulk polarization current $\vJ_{b}=\partial\vP/\partial t=-\omega P_{0}\sin\omega t\,\zhat$. A neutralizing surface charge alone does not make the source silent. It cancels the electric dipole but leaves $\vJ_{b}$, which still has an ordinary transverse radiation amplitude. To cancel the Maxwell channel completely, add the free compensating layer and volume current
\begin{equation}
\sigma_{f}^{\rm car}=-\sigma_{b},\qquad
\vJ_{f,\rm car}^{\,\mathrm{vol}}=-\pdt{\vP}=+\omega P_{0}\sin\omega t\,\zhat.
\label{eq:carriercomp}
\end{equation}
As distributions, the complete carrier has $\rho_{\rm car}=0$ and $\vJ_{\rm car}=\zero$. Its retarded source-generated potentials therefore vanish up to homogeneous solutions. Thus $\vE=\vB=\zero$ and $\langle P_{T}\rangle=0$ exactly. The free and bound carrier sectors are separately conservative after the standard surface-continuity term is included (Appendix~\ref{app:bound}). Hence $\Lambda_{\rm car}=0$, and the carrier is also silent in EED.

\subsection{Anomalous Scalar-Source Sector}

Opening the scalar channel requires a source that violates local charge conservation. A dipolar surface anomaly is placed on the compensating layer. It is the lowest multipole with zero net anomalous charge.
\begin{equation}
\lambda_{s}(\theta,t)=\lambda_{a}\cos(\omega t)\cos\theta,
\qquad
\Lambda_{a}(\vr,t)=\lambda_{s}(\theta,t)\,\delta(r-R).
\label{eq:surfaceanomaly}
\end{equation}
This source has no monopole and has a nonzero dipole. Here $\delta(r-R)$ is the one-dimensional radial delta function used with $d^{3}r=r^{2}dr\,d\Omega$. The integral $\int\Lambda_{a}\,d^{3}r$ vanishes at every instant, and the associated anomalous surface charge also integrates to zero. The model therefore violates local continuity while conserving total charge. The dipole is its leading radiating moment. The parameter $\lambda_{a}$ is the real time-domain amplitude. Under the convention of Sec.~\ref{sec:decomp}, the phasor amplitude is $\Lambda_{a,\omega}=\lambda_{a}\cos\theta\,\delta(r-R)$. One prescribed realization in the EED field equations is the anomalous surface-charge modulation
\begin{equation}
\sigma_{a}(\theta,t)=\frac{\lambda_{a}}{\omega}\sin(\omega t)\cos\theta,
\qquad \vJ_{a}=\zero.
\label{eq:minrealization}
\end{equation}
It satisfies $\partial\sigma_{a}/\partial t=\lambda_{s}$, so the surface continuity equation is violated by construction. Because the anomalous sector has no current, its on-shell transverse amplitude $\Pperp(\rhat)\calJ_{a,\omega}(\rhat)$ vanishes. It produces no ordinary transverse radiation and no scalar-channel magnetic field. Its nonzero $\Lambda_{a}$ still drives $C$ through \eqref{eq:scalarbox}. After this sector is added, the total source is neither the zero Maxwell source nor a Maxwell-consistent source. The transverse radiation amplitude remains zero, leaving the scalar channel isolated from an electromagnetic background.

\subsection{Radiated Power}

The anomaly dipole is
\begin{equation}
\vp_{\Lambda}=\oint_{r=R}\vr'\,\lambda_{s}\,dA=\tfrac{4\pi}{3}R^{3}\lambda_{a}\cos(\omega t)\,\zhat,
\qquad p_{\Lambda,0}=\tfrac{4\pi}{3}R^{3}\lambda_{a}.
\end{equation}
Substitution into the dipole formula \eqref{eq:Pdip} gives
\begin{equation}
\boxed{\;\langle P\rangle_{\mathrm{SLW}}\simeq\frac{\mu_{0}\,\omega^{2}\,p_{\Lambda,0}^{2}}{24\pi c}
=\frac{2\pi\,\mu_{0}\,\omega^{2}\,R^{6}\,\lambda_{a}^{2}}{27\,c}\;.}
\label{eq:Pelectret}
\end{equation}
The power scales as $\lambda_a^2$ and vanishes with the anomaly. A purely magnetization-based source has $\Lambda=0$ and does not excite this channel.

\subsection{Exact Shell Power}

Equation~\eqref{eq:Pelectret} is the long-wavelength result. The ideal shell anomaly remains a single multipole at every $kR$, so \eqref{eq:Pmaster} gives an exact expression. For $\Lambda_{a,\omega}=\lambda_{a}\cos\theta\,\delta(r-R)$, only the $\ell=1,m=0$ moment remains. Appendix~\ref{app:electret} gives $q_{10,\omega}=\lambda_{a}R^{2}j_{1}(kR)\sqrt{4\pi/3}$. Therefore
\begin{align}
\langle P\rangle_{\mathrm{shell}}
&=\frac{c\mu_{0}}{2}|q_{10,\omega}|^{2}
=\frac{2\pi c\mu_{0}}{3}\lambda_a^2R^4j_1^2(kR),
\label{eq:Pshell-exact}\\
\langle P\rangle_{\mathrm{shell}}
&\simeq\frac{2\pi\mu_{0}\omega^2R^6\lambda_a^2}{27c}
\left[1-\frac{(kR)^2}{5}+O((kR)^4)\right],\qquad kR\ll1.
\label{eq:Pshell-small}
\end{align}
Here $j_1(x)=x/3-x^3/30+O(x^5)$ was used in the last step. The leading term is \eqref{eq:Pelectret}. The zeros of $j_{1}(kR)$ make the ideal shell radiation-silent. These zeros rely on an infinitesimally thin, uniformly phased layer. For a finite-thickness anomaly, $j_{1}(kR)$ is replaced by a radial average through the layer. Thickness, phase variation, material response, and drive nonuniformity shift the nulls and generally make the power nonzero at the thin-shell zeros, although the modified form factor may have other zeros. The nulls are therefore not universal. The dipole result is the first term of the exact shell power.

\subsection[Comparison with the Tensor Treatment of Minotti and Modanese]{Comparison with the Tensor Treatment of Ref.~\cite{MinottiModanese2021}}

The shell in Sec.~\ref{sec:electret} has $\vJ_{a}=\zero$ and a time-dependent, globally neutral dipolar surface charge. It belongs to the class of currentless oscillating dipoles. Reference~\cite{MinottiModanese2021} studies this class with the symmetric energy-momentum tensor derived there and obtains a different result. In the long-wavelength dipole limit, its cycle-averaged tensor flux is negative, $\langle W_{\mathrm{tensor}}\rangle=-\frac{\mu_{0}}{12\pi c}\langle|d\vp_{\Lambda}/dt|^{2}\rangle$ in the notation used here. The anomaly dipole $\vp_{\Lambda}=\int\vr\,\Lambda\,d^{3}r$ is written there as $\int\mathbf{x}\,I\,d^{3}x$. The tensor flux points inward. Since $\vJ=\zero$, the $\vJ\cdot\vE$ exchange channel is also absent. That treatment finds no path for locally transferred field energy to return or dissipate. The authors conclude that their particular two-charge currentless dipole is not physically possible, while noting that this does not rule out every similar source.
\\

The difference is explicit. At leading dipole order, the two results are
\[
\langle P\rangle_{\text{dipole}}=+\frac{\mu_{0}}{12\pi c}\Big\langle\Big|\frac{d\vp_{\Lambda}}{dt}\Big|^{2}\Big\rangle,
\qquad
\langle W_{\mathrm{tensor}}\rangle=-\frac{\mu_{0}}{12\pi c}\Big\langle\Big|\frac{d\vp_{\Lambda}}{dt}\Big|^{2}\Big\rangle.
\]
The two expressions have the same magnitude and opposite signs for the same currentless anomaly dipole. The field-only balance gives outward flow, while the tensor balance gives inward flow. Reference~\cite{MinottiModanese2021} establishes the tensor result at leading dipole order. The corresponding all-orders tensor flux for the shell is not calculated here. The expressions come from different currents, so there is no mathematical contradiction. They cannot both represent the same physical radiated power, and this paper does not decide which one does. The shell is therefore a prescribed-source illustration of the multipole formulas, not proof that such a device would radiate. The tensor comparison also leaves the physical admissibility of a currentless anomalous source unresolved. The field decomposition and no-go results in Secs.~\ref{sec:decomp} through~\ref{sec:corollaries} do not depend on that question.

\subsection{Source Exchange in the Adopted Balance}

In the field-only balance \eqref{eq:poynting}, integrated over a region enclosing the shell and averaged over one cycle, the volume integral of the source term $c^{2}\rho C$ equals the outgoing scalar power, since $\vJ_{a}=\zero$ and the cycle-averaged stored energy is stationary. Here $\rho$ is the anomalous surface charge associated with $\sigma_{a}$. This equality is an integrated, cycle-averaged consequence of \eqref{eq:poynting}. It is not a pointwise identity and does not determine a microscopic source energy budget. For the shell,
$\rho_{a,\omega}(\vr)=\sigma_{a,\omega}(\theta)\delta(r-R)$ with
$\sigma_{a,\omega}=(i\lambda_a/\omega)\cos\theta$, and the partial-wave Green function gives $C_{\omega}(R,\theta)=i\mu_{0}k\lambda_{a}R^{2}j_{1}(kR)\,h_{1}^{(1)}(kR)\cos\theta$. The source integral is
\[
\tfrac12\,\mathrm{Re}\!\int c^{2}\rho_{a,\omega}C_{\omega}^{*}\,d^{3}r
=\tfrac12\,c^{2}\mu_{0}\frac{k}{\omega}\lambda_{a}^{2}R^{4}j_{1}^{2}(kR)\cdot\frac{4\pi}{3}
=\frac{2\pi c\mu_{0}}{3}\lambda_{a}^{2}R^{4}j_{1}^{2}(kR),
\]
where $\mathrm{Re}[\,j_{1}h_{1}^{(1)*}]=j_{1}^{2}$. The result is exactly \eqref{eq:Pshell-exact}. The surface value is well defined because $C$ is continuous across the layer. Its radial derivative has the jump $[\partial_{r}C_{\omega}]_{r=R}\equiv(\partial_{r}C_{\omega})_{R^{+}}-(\partial_{r}C_{\omega})_{R^{-}}=-\mu_{0}\lambda_{s,\omega}(\theta)$, with $\lambda_{s,\omega}(\theta)=\lambda_{a}\cos\theta$, obtained by integrating the scalar Helmholtz equation across the shell.

\begin{figure}[t]
\centering
\includegraphics[width=0.86\textwidth]{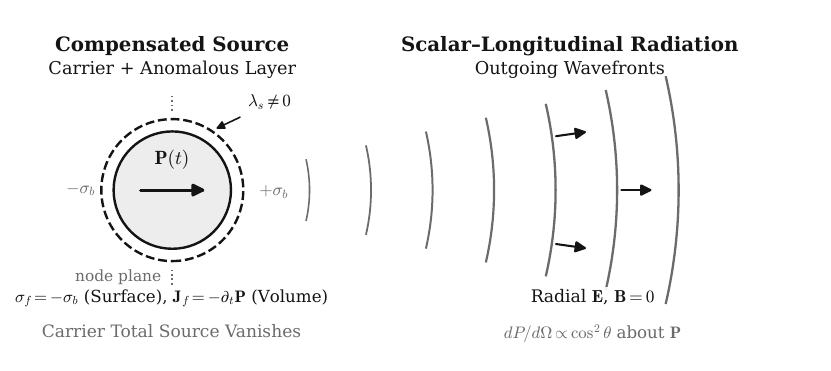}
\caption{Compensated-source bookkeeping. The carrier has $\rho_{\rm car}=0$ and $\vJ_{\rm car}=\zero$ as distributions. The anomalous layer has $\lambda_s\neq0$ and $\Pperp(\rhat)\calJ_{a,\omega}(\rhat)=0$, so it opens only the scalar-longitudinal channel. Dotted guides mark the equatorial node. The arcs are schematic wavefronts, not the angular pattern.}
\label{fig:electret}
\end{figure}

\FloatBarrier
\section{Comparison with Classical Electrodynamics and Illustrative Field Scaling}
\label{sec:compare}

Table~\ref{tab:compare} summarizes the contrast with Maxwell theory. The scalar sector is inactive in Maxwell theory and in EED with conserved sources. It can be activated only when a free source violates local charge conservation, and it radiates only when at least one on-shell anomaly moment $q_{\ell m,\omega}$ is nonzero.

\begin{table}[t]
\caption{Radiation-zone content in three cases. For zero scalar initial data, the scalar channel is absent unless the free charge-current pair violates local continuity, $\Lambda=\partial\rho/\partial t+\dive\vJ\neq0$. The transverse Maxwell sector is unchanged by $C$. In the last column, $\Lambda\neq0$ is necessary but not sufficient because at least one on-shell moment $q_{\ell m,\omega}$ must also be nonzero}
\label{tab:compare}
\centering
\footnotesize
\begin{tabularx}{\textwidth}{@{}p{0.20\textwidth}YYY@{}}
\toprule
& \textbf{Classical} & \textbf{EED, $\Lambda=0$} & \textbf{EED, $\Lambda\neq0$}\\
\midrule
Scalar Field $C$ & Not Dynamical (gauge only) & $\equiv0$ (zero initial data) & Generically $\neq0$\\
Radiative $E_{L}$ & None & None & $\neq0$ if radiative $q_{\ell m,\omega}\neq0$\\
Magnetic Field $\vB$ & Transverse Radiation & Transverse Radiation & $\vB_{L}=\zero$, $\vB_{T}$ if $\Pperp(\rhat)\calJ_{\omega}(\rhat)\neq0$\\
Transverse Radiation & From $\Pperp(\rhat)\calJ_{\omega}(\rhat)$ & As in Maxwell & From $\Pperp(\rhat)\calJ_{\omega}(\rhat)$, unchanged\\
Scalar Radiation & Absent & None & $\tfrac12 c\mu_{0}\sum|q_{\ell m,\omega}|^{2}$ if $q_{\ell m,\omega}\neq0$\\
Source Condition & $\Lambda=0$ & $\Lambda=0$ & $\Lambda\neq0$ with $q_{\ell m,\omega}\neq0$\\
\botrule
\end{tabularx}
\end{table}

\subsection{Wave Structure}

Figure~\ref{fig:structure} shows the field geometry. Maxwell radiation is transverse, so $\vE$ and $\vB$ are perpendicular to the propagation direction. In the EED branch, the electric field is longitudinal, $\vB=\zero$, and $C$ travels in phase with $\vE_{L}$. The missing magnetic field is a defining property of the mode, not a small correction.

\subsection{Angular Pattern}

Figure~\ref{fig:patterns} compares dipole patterns. A Maxwell electric dipole comes from a conserved charge-current pair and has a $\sin^{2}\theta$ pattern, strongest across the source axis. The scalar dipole comes from $\vp_{\Lambda}$ and has a $\cos^{2}\theta$ pattern, strongest along the axis. The polar and Cartesian plots show the same complementary behavior.

\subsection{Multipole Hierarchy and Power}

Figure~\ref{fig:power} shows the long-wavelength scaling $\langle P_{\ell}\rangle\propto x^{2\ell}$ for scalar multipoles of comparable reduced strength, with $x=kR$. Each higher multipole is suppressed by another factor of $x^{2}$, so the lowest nonzero multipole dominates. The compensated electret anomaly has no monopole and begins at dipole order. Figure~\ref{fig:shell} gives the exact shell power from \eqref{eq:Pshell-exact}, including the finite-$kR$ behavior and the form-factor zeros of the ideal shell. The compensated carrier has exactly zero Maxwell radiation.

\begin{figure}[t]
\centering
\includegraphics[width=0.94\textwidth]{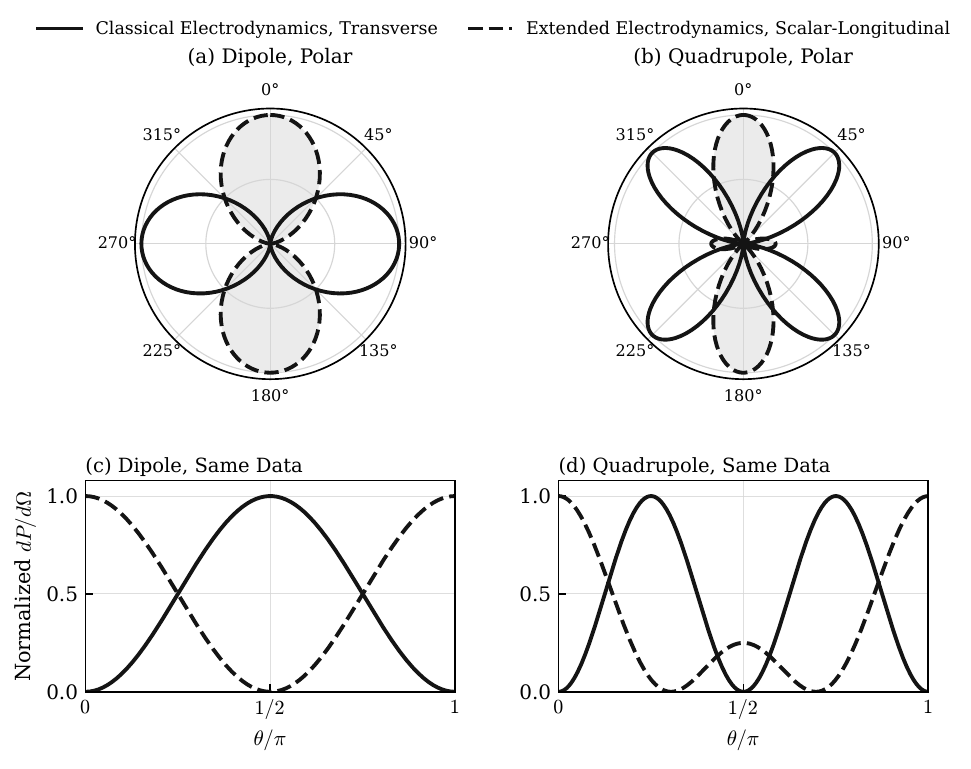}
\caption{Normalized angular patterns for transverse Maxwell radiation and EED scalar-longitudinal radiation. The scalar dipole is proportional to $\cos^2\theta$, whereas the Maxwell electric dipole is proportional to $\sin^2\theta$. For the axisymmetric quadrupole, the corresponding curves are proportional to $(3\cos^2\theta-1)^2$ and $\sin^2\theta\cos^2\theta$.}
\label{fig:patterns}
\end{figure}

\begin{figure}[t]
\centering
\includegraphics[width=0.74\textwidth]{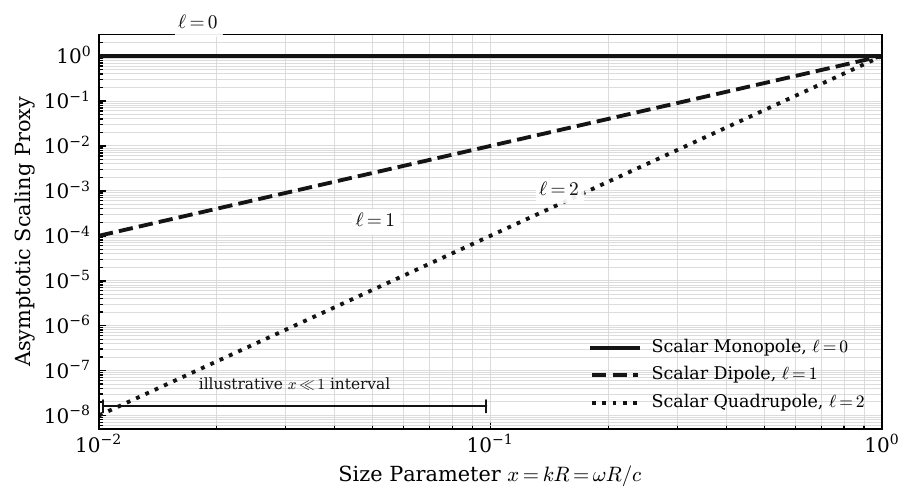}
\caption{Long-wavelength scaling proxies $1$, $x^2$, and $x^4$ for scalar monopole, dipole, and quadrupole terms, with $x=kR$. Each curve is normalized at $x=1$ for display, although the asymptotic approximation requires $x\ll1$. Relative powers also contain $[(2\ell+1)!!]^{-2}$ and the source moment norms.}
\label{fig:power}
\end{figure}

\begin{figure}[t]
\centering
\includegraphics[width=0.62\textwidth]{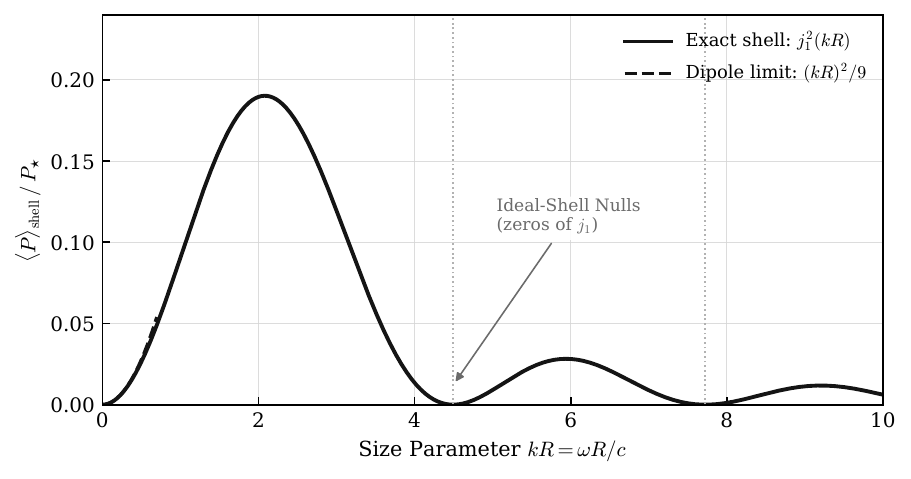}
\caption{Exact shell power from \eqref{eq:Pshell-exact}, normalized by $P_{\star}=2\pi c\mu_{0}\lambda_a^2R^4/3$. The dashed curve is the dipole limit $(kR)^2/9$. The exact form factor vanishes at the zeros of $j_1$. Finite thickness or nonuniform phase shifts these nulls.}
\label{fig:shell}
\end{figure}

\subsection{Observational Context}

For a fixed normalized anomaly profile, $\langle P_{L}\rangle$ scales with the square of the overall anomaly amplitude. Define $\xi\equiv\lambda_{a}/J_{\mathrm{ref}}$ as a convenient dimensionless parameter, not as a universal measure of charge nonconservation. Here $J_{\mathrm{ref}}$ is a reference current-density scale in $\mathrm{A/m^{2}}$. A surface anomaly is of order $\lambda_{a}\sim\xi J_{\mathrm{ref}}$ in $\mathrm{A/m^{2}}$, while a volume anomaly is of order $\Lambda\sim\xi J_{\mathrm{ref}}/R$ in $\mathrm{A/m^{3}}$. In either case, \eqref{eq:Pelectret} gives $\langle P\rangle_{\mathrm{SLW}}\propto\xi^{2}$. A fundamental violation of charge conservation, such as the broken-gauge and brane-world models of Refs.~\cite{Ignatiev1979,Dubovsky2000}, is constrained by searches for specific microscopic processes. Examples include the electron-lifetime bound $\tau\gtrsim6.6\times10^{28}\,$yr for $e\to\nu\gamma$~\cite{Borexino2015}, the Majorana germanium-detector search~\cite{Majorana2024}, and the tests reviewed in Ref.~\cite{Okun1989}. These process-specific limits do not directly constrain the coherent ratio $\xi$ without a microscopic mapping. Reduced transport models can also contain an effective anomaly even when the underlying theory conserves charge~\cite{Modanese2017,MinottiModanese2021,MinottiModanese2021b,MinottiModanese2022}. The manuscript therefore reports only the scaling $\langle P\rangle_{\mathrm{SLW}}\propto\xi^2$.

\subsection{Benchmark Magnitudes}

As an illustration, take $R=1\,\mathrm{cm}$ and $\omega/2\pi=1\,\mathrm{GHz}$. Then $x=kR\approx0.21$, which is in the dipole range. At this value, the exact shell power \eqref{eq:Pshell-exact} is $j_{1}^{2}(kR)/[(kR)^{2}/9]\approx0.991$ times the dipole result, a $0.9\%$ correction. Also take $J_{\mathrm{ref}}=10^{6}\,\mathrm{A/m^{2}}$ as an illustrative normalization, not as a value for a specific conductor, and set $\lambda_{a}=\xi J_{\mathrm{ref}}$. Equation~\eqref{eq:Pelectret} gives $\langle P\rangle_{\mathrm{SLW}}\approx3.9\times10^{4}\,\xi^{2}\,\mathrm{W}$. Holding $\lambda_{a}$ fixed gives the small-$kR$ scaling $\omega^{2}\lambda_{a}^{2}$. Holding the anomalous surface-charge amplitude $\sigma_{a,0}=\lambda_{a}/\omega$ fixed instead gives $\omega^{4}\sigma_{a,0}^{2}$. The on-axis longitudinal field at distance $r$ follows from $\langle S_{r}\rangle(\theta{=}0)=3\langle P\rangle_{\mathrm{SLW}}/4\pi r^{2}$ and $|E_{L,\omega}|=\sqrt{2\mu_{0}c\,\langle S_{r}\rangle}$. It is
\begin{equation}
|E_{L,\omega}|=\frac{1}{r}\sqrt{\frac{3\mu_{0}c\,\langle P\rangle_{\mathrm{SLW}}}{2\pi}}
\;\approx\;2.6\times10^{3}\,\xi\ \ \mathrm{V/m}\qquad(r=1\,\mathrm{m}).
\label{eq:benchmark}
\end{equation}
This is the peak phasor amplitude under the convention of Sec.~\ref{sec:decomp}. The root-mean-square value is smaller by $\sqrt{2}$. In this normalization, $\xi\sim10^{-9}$ gives a coherent source-locked longitudinal field of a few $\mu\mathrm{V/m}$ at one meter and an adopted-current flux of about $40\,\mathrm{fW}$. These values normalize the prescribed source model and do not establish detectability. Corollary~\ref{cor:detector} shows that reception depends on the detector coupling, and $\xi$ remains model dependent.

\subsection{Possible Experimental Signatures}

If such a process exists, three ideal features provide consistency checks rather than unique experimental discriminators.
\\
\begin{enumerate}[label=(\roman*),leftmargin=*,itemsep=2pt,topsep=3pt]
\item The pure scalar-longitudinal channel has no magnetic field of its own. An ideal inductive pickup therefore has no direct Faraday signal, although real receivers may couple capacitively or mechanically.
\\
\item A linearly polarized scalar dipole has a $\cos^2\theta$ pattern, strongest on its axis where a co-aligned Maxwell electric dipole has a node.
\item The mode travels at $c$ and satisfies $C_0=E_0/c$. A separate measurement of $C$ requires a specified $C$-sensitive coupling~\cite{HivelyGiakos2012,ReedHively2020}.
\\
\end{enumerate}

Mesoscopic transport has been proposed as a setting for an effective $\Lambda\neq0$ in reduced descriptions~\cite{Modanese2017,MinottiModanese2021,MinottiModanese2021b,MinottiModanese2022}. In particular, approximate non-equilibrium Green-function and density-functional treatments can be compatible with a modified continuity equation $\partial_t\langle\rho\rangle+(1-\gamma)\dive\langle\vJ\rangle=0$~\cite{MinottiModanese2025}. Fractional or nonlocal Schr\"odinger models provide other examples in which the conventional current requires an additional source term~\cite{Wei2016,Modanese2018}. Related fields have been calculated in the secondary-current formulation~\cite{Modanese2017b}. Treating such a reduced-model defect as the physical EED source is an extra phenomenological assumption. The microscopic theory, the reduced description, and the EED source model must therefore be kept distinct. Proposed detector concepts appear in Refs.~\cite{MinottiModanese2023,MinottiModanese2024}.

\subsection{No Maxwell Counterpart}

No charge-conserving free-space Maxwell source produces the far-zone combination of longitudinal $\vE_{L}$, no scalar-channel magnetic field $\vB_{L}=\zero$, and a $\cos^{2}\theta$ dipole pattern. In Maxwell theory, $\partial\rho/\partial t+\dive\vJ=0$ follows from the field equations, so a genuinely nonconserving source has no Maxwell solution. A conserved source, including one whose apparent discontinuity disappears when a missing return path is restored, has $\Lambda=0$, $C\equiv0$, and no scalar emission. Classical longitudinal electric fields can still occur in near zones, plasmas, and guided structures. An experiment would therefore need to operate in the far zone and exclude near-field pickup, plasma or guided-mode leakage, and instrumental artifacts. After those effects were ruled out, the three ideal features would be consistent with an added longitudinal scalar response. Turning a null result into a bound would require a calibrated source-to-multipole transfer function, a detector model, and a statistical background analysis. None is attempted here.

\FloatBarrier
\section{Discussion and Conclusions}\label{sec:discussion}

\subsection{Scope of the Claims}
The results are conditional consequences of Eqs.~\eqref{eq:potwave}--\eqref{eq:scalarbox} for prescribed sources and the stated initial and boundary conditions. They do not imply that charge conservation fails in nature~\cite{Okun1989,Borexino2015,Majorana2024}. The field decomposition and multipole expansion follow from the retarded equations. Power statements additionally use the field-only balance \eqref{eq:uS}. The compensated electret is a diagnostic source, not a model of an ordinary electret.

\subsection{Poynting Functional and the Unresolved Energy Sign}
The scalar-longitudinal sector is related to the component removed from the physical state space in Gupta-Bleuler quantization~\cite{Gupta1950,Bleuler1950}. Equation~\eqref{eq:uS} gives a positive local density for the free wave, but it does not determine the canonical energy. For example, $\Phi=\Phi_0\cos(kz-\omega t)$ and $\vA=\zero$ give the same mode as \eqref{eq:slw}, shifted by one quarter period, while
$\mathcal H_{\mathrm{can}}=-\tfrac12\eps[|\grad\Phi|^2+c^{-2}(\partial_t\Phi)^2]<0$ for this infinite plane-wave representative. An integrated Noether charge requires finite-energy data and specified surface terms. The relation among \eqref{eq:uS}, the canonical current, and a Hilbert stress tensor remains unresolved. The tensor treatment of Ref.~\cite{MinottiModanese2021} can give a different radiated-flux sign. Thus \eqref{eq:Pmaster} is the outgoing flux of the adopted current, not by itself a proof of positive physical energy.

\subsection{Relation to Earlier Work}
The extended framework and its longitudinal solutions trace to Refs.~\cite{Ohmura1956,AharonovBohm1963,vanVlaenderen2001,HivelyGiakos2012}. Woodside studied related uniqueness questions~\cite{Woodside1999,Woodside2000,Woodside2009}. Other developments include Refs.~\cite{JimenezMaroto2009,JimenezMaroto2011,KellerHively2019,ReedHively2020,HivelyLand2021}. Anomalous sources, radiation, material response, and proposed detection have been studied in Refs.~\cite{Modanese2017,MinottiModanese2021,MinottiModanese2021b,MinottiModanese2022,MinottiModanese2023,MinottiModanese2023b,MinottiModanese2024,Modanese2017b}. The present analysis adds a general radiation-zone channel decomposition, the exact scalar multipole formula \eqref{eq:Pmaster}, and the bound-source and detector corollaries derived from them.

\subsection{Validity and Outlook}
Equation~\eqref{eq:Pmaster} is exact for a localized harmonic anomaly, while \eqref{eq:Pell}, \eqref{eq:Pmono}, \eqref{eq:Pdip}, and \eqref{eq:Pelectret} are long-wavelength results. Equation~\eqref{eq:Pshell-exact} is exact for the ideal shell at every $kR$. The carrier cancellation is exact because its total source vanishes as a distribution. Whether a physical system produces an effective $\Lambda\neq0$ is outside this work. Within EED, a radiative anomaly can produce a longitudinal electric field traveling at $c$, with no magnetic field of its own and a $\cos^2\theta$ dipole pattern. Conserved matter with no incoming or homogeneous scalar field produces none of these source-generated fields.

\FloatBarrier
\backmatter
\begin{appendices}

\section{Derivation of the Extended Poynting Theorem}
\label{app:poynting}
Start with \eqref{eq:eedfields}, the homogeneous equations $\curl\vE=-\partial\vB/\partial t$ and $\dive\vB=0$, and the identity $\dive(\vE\times\vB)=\vB\cdot(\curl\vE)-\vE\cdot(\curl\vB)$. Substituting $\curl\vB=\mu_{0}\vJ+\grad C+c^{-2}\partial\vE/\partial t$ and $\curl\vE=-\partial\vB/\partial t$ gives
\begin{equation}
-\frac{1}{\mu_{0}}\dive(\vE\times\vB)=\vE\cdot\vJ+\frac{1}{\mu_{0}}\vE\cdot\grad C +\frac{\partial}{\partial t}\!\left(\frac{\eps}{2}|\vE|^{2}+\frac{1}{2\mu_{0}}|\vB|^{2}\right),
\end{equation}
where $c^{-2}/\mu_{0}=\eps$. Next use $\vE\cdot\grad C=\dive(C\vE)-C\dive\vE$ and the extended Gauss law
$\dive\vE=\rho/\eps-\partial C/\partial t$. Then $\mu_{0}^{-1}\vE\cdot\grad C=\mu_{0}^{-1}\dive(C\vE)-c^{2}\rho C+\partial(C^{2}/2\mu_{0})/\partial t$, where $1/(\mu_{0}\eps)=c^{2}$. Collecting terms gives
\begin{equation}
\frac{\partial u}{\partial t}+\dive\vS=-\vE\cdot\vJ+c^{2}\rho C,
\label{eq:poynting}
\end{equation}
with $u$ and $\vS$ defined in \eqref{eq:uS}. The scalar terms are the energy density $C^{2}/2\mu_{0}$, the flux $C\vE/\mu_{0}$, and the source-exchange term $c^{2}\rho C$.

\section{Far-Zone Fields from the Retarded Potentials}
\label{app:farfield}
Equations~\eqref{eq:farC} through~\eqref{eq:farEL} follow directly from the retarded solutions of \eqref{eq:potwave}. In the frequency domain,
\begin{equation}
\Phi_{\omega}(\vr)=\frac{1}{4\pi\eps}\int \rho_{\omega}(\vr')\,\frac{e^{ik|\vr-\vr'|}}{|\vr-\vr'|}\,d^{3}r',\qquad
\vA_{\omega}(\vr)=\frac{\mu_{0}}{4\pi}\int \vJ_{\omega}(\vr')\,\frac{e^{ik|\vr-\vr'|}}{|\vr-\vr'|}\,d^{3}r'.
\label{eq:retpots}
\end{equation}
Let $r\to\infty$ while the source remains inside $r'\le R$. Then $|\vr-\vr'|=r-\rhat\cdot\vr'+O(r'^{2}/r)$. The kernel becomes $e^{ik|\vr-\vr'|}/|\vr-\vr'|=(e^{ikr}/r)\,e^{-ik\rhat\cdot\vr'}\,[1+O(1/r)]$, the standard far-zone expansion of Ref.~\cite{Jackson}, so the potentials reduce to the on-shell transforms in \eqref{eq:on-shell-rhoJ}.
\begin{equation}
\Phi_{\omega}=\frac{e^{ikr}}{4\pi\eps\,r}\,\calR_{\omega}(\rhat)+O(r^{-2}),\qquad
\vA_{\omega}=\frac{\mu_{0}e^{ikr}}{4\pi r}\,\calJ_{\omega}(\rhat)+O(r^{-2}).
\label{eq:farpots}
\end{equation}
For an outgoing wave, the leading gradient is radial, $\grad\big[f(\rhat)\,e^{ikr}/r\big]=ik\,\rhat\,f(\rhat)\,e^{ikr}/r+O(r^{-2})$. Angular derivatives begin at $O(r^{-2})$. The four far-zone amplitudes follow.

\subsection{Scalar Field}

Using $\partial_{t}\to-i\omega$ in the definition of $C$,
\begin{equation}
\begin{aligned}
C_{\omega}
&=\dive\vA_{\omega}-\frac{i\omega}{c^{2}}\,\Phi_{\omega}\\
&=\frac{e^{ikr}}{4\pi r}\Big[ik\mu_{0}\,\rhat\cdot\calJ_{\omega}(\rhat)-\frac{i\omega}{c^{2}\eps}\,\calR_{\omega}(\rhat)\Big]+O(r^{-2})\\
&=\frac{\mu_{0}e^{ikr}}{4\pi r}\,\big[ik\,\rhat\cdot\calJ_{\omega}(\rhat)-i\omega\,\calR_{\omega}(\rhat)\big]+O(r^{-2}),
\end{aligned}
\end{equation}
where $1/(c^{2}\eps)=\mu_{0}$. This is \eqref{eq:farC}, with $\Lambda_{\mathrm{on},\omega}$ given by \eqref{eq:on-shell-Lambda}.

\subsection{Electric Field}

From $\vE_{\omega}=-\grad\Phi_{\omega}+i\omega\vA_{\omega}=\frac{e^{ikr}}{4\pi r}\big[-\frac{ik}{\eps}\,\rhat\,\calR_{\omega}(\rhat)+i\omega\mu_{0}\,\calJ_{\omega}(\rhat)\big]+O(r^{-2})$, the transverse projector removes the radial $\grad\Phi_{\omega}$ term.
\begin{equation}
\Pperp(\rhat)\vE_{\omega}=\frac{i\omega\mu_{0}e^{ikr}}{4\pi r}\,\Pperp(\rhat)\calJ_{\omega}(\rhat)+O(r^{-2}).
\end{equation}
This is \eqref{eq:farET}. The radial projection keeps both terms.
\begin{equation}
\begin{aligned}
\rhat\cdot\vE_{\omega}
&=\frac{e^{ikr}}{4\pi r}\Big[-\frac{ik}{\eps}\,\calR_{\omega}(\rhat)+i\omega\mu_{0}\,\rhat\cdot\calJ_{\omega}(\rhat)\Big]\\
&=\frac{\mu_{0}c\,e^{ikr}}{4\pi r}\,\big[-i\omega\,\calR_{\omega}(\rhat)+ik\,\rhat\cdot\calJ_{\omega}(\rhat)\big]\\
&=c\,C_{\omega}+O(r^{-2}).
\end{aligned}
\label{eq:radialE}
\end{equation}
Using $1/\eps=\mu_{0}c^{2}$ and $\omega=ck$ gives \eqref{eq:farEL}. The radiation-zone radial electric field is exactly $c$ times the scalar amplitude. This is the far-field form of $C_{0}=E_{0}/c$ in \eqref{eq:slw}. The bracket in \eqref{eq:radialE} is proportional to $\Lambda_{\mathrm{on},\omega}$ and vanishes for a conserved source. The usual transverse nature of Maxwell radiation is therefore the $\Lambda=0$ limit of the same calculation.

\subsection{Magnetic Field}

\begin{align}
\vB_{\omega}
&=\curl\vA_{\omega}=ik\,\rhat\times\vA_{\omega}+O(r^{-2}),\notag\\
&=\frac{ik\mu_{0}e^{ikr}}{4\pi r}\,\rhat\times\calJ_{\omega}(\rhat)+O(r^{-2})\notag\\
&=\frac{1}{c}\,\rhat\times\vE_{T,\omega}+O(r^{-2}).
\end{align}
This is \eqref{eq:farBT} because $\rhat\times\calJ_{\omega}(\rhat)=\rhat\times\Pperp(\rhat)\calJ_{\omega}(\rhat)$ and $\omega=ck$. The omitted $O(1/r)$ kernel correction and angular derivatives contribute to the radial flux only at $O(r^{-3})$ and $O(r^{-4})$. These are the terms discarded in the proof of Theorem~\ref{thm:decomp}.

\section{Scalar Multipole Power Formula}
\label{app:multipole}
For a harmonic source,
$\Lambda(\vr,t)=\tfrac12[\Lambda_{\omega}(\vr)e^{-i\omega t}+\Lambda_{\omega}^{*}(\vr)e^{i\omega t}]$,
and likewise for $C$. The outgoing solution of
$(\nabla^2+k^2)C_{\omega}=-\mu_{0}\Lambda_{\omega}$ is
$C_{\omega}=\mu_{0}\int G\Lambda_{\omega}\,d^3r'$, with
$G=e^{ik|\vr-\vr'|}/(4\pi|\vr-\vr'|)$. Expanding
\[
G=ik\sum_{\ell m}j_{\ell}(kr_<)h_{\ell}^{(1)}(kr_>)Y_{\ell m}(\rhat)Y_{\ell m}^{*}(\rhat')
\]
gives, outside the source,
\begin{equation}
\begin{aligned}
C_{\omega}(\vr)&=i\mu_{0}k\sum_{\ell m}q_{\ell m,\omega}h_{\ell}^{(1)}(kr)Y_{\ell m}(\rhat),\\
q_{\ell m,\omega}&=\int j_{\ell}(kr')Y_{\ell m}^{*}(\rhat')\Lambda_{\omega}(\vr')\,d^3r'.
\end{aligned}
\label{eq:app-multipole}
\end{equation}
Using $h_{\ell}^{(1)}(kr)\to(-i)^{\ell+1}e^{ikr}/(kr)$ and
$\langle S_r\rangle=c|C_{\omega}|^2/(2\mu_{0})$ gives
\begin{equation}
\langle P_L\rangle=\frac{c\mu_{0}}{2}\sum_{\ell m}|q_{\ell m,\omega}|^2.
\end{equation}
For $kR\ll1$,
\begin{equation}
q_{\ell m,\omega}\simeq\frac{k^{\ell}}{(2\ell+1)!!}\mathcal M_{\ell m,\omega},
\qquad
\langle P_{\ell}\rangle\simeq
\frac{c\mu_{0}}{2}\frac{k^{2\ell}}{[(2\ell+1)!!]^2}
\sum_m|\mathcal M_{\ell m,\omega}|^2.
\end{equation}
For $\ell=1$,
$Y_{10}=\sqrt{3/(4\pi)}\cos\theta$ and
$Y_{1,\pm1}=\mp\sqrt{3/(8\pi)}\sin\theta e^{\pm i\phi}$. Hence
$\mathcal M_{10,\omega}=\sqrt{3/(4\pi)}p_{z,\omega}$ and
$\mathcal M_{1,\pm1,\omega}=\mp\sqrt{3/(8\pi)}(p_{x,\omega}\mp ip_{y,\omega})$, so
$\sum_m|\mathcal M_{1m,\omega}|^2=(3/4\pi)|\vp_{\Lambda,\omega}|^2$.

\section{Distributional Conservation of Bound and Carrier Sources}
\label{app:bound}
For $\rho_{b}=-\dive\vP$ and $\vJ_{b}=\partial\vP/\partial t+\curl\vM$, the bulk anomaly is $\Lambda_{b}=\partial\rho_{b}/\partial t+\dive\vJ_{b}=-\dive(\partial\vP/\partial t)+\dive(\partial\vP/\partial t)+\dive(\curl\vM)=0$. At a surface $\Sigma$ where $\vP$ is discontinuous, the singular terms give $\sigma_{b}=\nhat\cdot(\vP_{1}-\vP_{2})$ and $\vK_{b}=(\vM_{1}-\vM_{2})\times\nhat$. The normal $\nhat$ points from region 1 to region 2. At a medium-vacuum interface, $\vP_{2}=\vM_{2}=\zero$. The surface part of $\partial_{t}\rho_{b}+\dive\vJ_{b}=0$ is
\begin{equation}
\partial_{t}\sigma_{b}+\nabla_{\!s}\!\cdot\vK_{b}+\nhat\cdot\big(\vJ_{b,2}-\vJ_{b,1}\big)=0.
\end{equation}
The normal jump of the polarization current supplies $\partial_{t}\sigma_{b}=\nhat\cdot(\partial_{t}\vP_{1}-\partial_{t}\vP_{2})$. The identity $\dive(\curl\vM)\equiv0$ accounts for the magnetization sheet. A discontinuous polarization, including the electret sphere, therefore does not generate $\Lambda$ and cannot excite the scalar channel.

\subsection{Free Compensating Carrier}

The free part of the carrier in Sec.~\ref{sec:electret} is also conserved as a distribution. In the bulk, $\rho_{f}=0$ and $\vJ_{f,\rm car}^{\,\mathrm{vol}}=\omega P_{0}\sin\omega t\,\zhat$ is uniform for $r<R$ and zero outside. Thus $\partial_{t}\rho_{f}+\dive\vJ_{f}=0$ away from the surface. At $r=R$, the surface charge is $\sigma_{f}^{\rm car}=-P_{0}\cos\omega t\cos\theta$, there is no free surface current, and the outward normal jump is $\nhat\cdot(\vJ_{\rm out}-\vJ_{\rm in})=-\omega P_{0}\sin\omega t\,(\zhat\cdot\rhat)=-\omega P_{0}\sin\omega t\cos\theta$. The surface continuity equation is
\begin{equation}
\partial_{t}\sigma_{f}^{\rm car}+\nhat\cdot\big(\vJ_{\rm out}-\vJ_{\rm in}\big)
=\omega P_{0}\sin\omega t\cos\theta-\omega P_{0}\sin\omega t\cos\theta=0.
\end{equation}
Thus $\Lambda_{\rm car}=0$ distributionally, confirming that the complete carrier is silent in both EED and Maxwell theory.

\section{Electret Surface Anomaly and Radiated Power}
\label{app:electret}
The anomalous surface layer introduced in Sec.~\ref{sec:electret} is
\begin{equation}
\Lambda_{a}(\vr,t)=\lambda_{a}\cos(\omega t)\cos\theta\,\delta(r-R).
\end{equation}
One EED realization has $\vJ_{a}=\zero$ and $\sigma_{a}(\theta,t)=(\lambda_{a}/\omega)\sin(\omega t)\cos\theta$ at $r=R$. Then $\partial_{t}\sigma_{a}=\lambda_{a}\cos(\omega t)\cos\theta$, so the surface continuity equation is violated by construction. Since $\calJ_{a,\omega}(\rhat)=0$, the transverse amplitude \eqref{eq:farET} vanishes. The scalar amplitude \eqref{eq:farC} is set by the on-shell moments of $\Lambda_{a,\omega}$ and is nonzero except at the ideal-shell zeros of $j_{1}(kR)$.
\\

The anomaly dipole is
\begin{equation}
\vp_{\Lambda}=\oint_{r=R}\vr'\,\lambda_{s}\,dA =\lambda_{a}\cos(\omega t)R^{3}\!\oint\rhat\,\cos\theta\,d\Omega =\frac{4\pi}{3}R^{3}\lambda_{a}\cos(\omega t)\,\zhat,
\end{equation}
because $\oint\rhat\cos\theta\,d\Omega=\zhat\!\oint\cos^{2}\theta\,d\Omega=(4\pi/3)\zhat$. Thus $p_{\Lambda,0}=(4\pi/3)R^{3}\lambda_{a}$ and the monopole is zero. Substitution into $\langle P\rangle=\frac{\mu_{0}\omega^{2}p_{\Lambda,0}^{2}}{24\pi c}$ from Appendix~\ref{app:multipole}, with $\ell=1$ and $\langle|d\vp_{\Lambda}/dt|^{2}\rangle=\tfrac12\omega^{2}p_{\Lambda,0}^{2}$, gives
\begin{equation}
\langle P\rangle_{\mathrm{SLW}}=\frac{\mu_{0}\omega^{2}}{24\pi c}\left(\frac{4\pi}{3}R^{3}\lambda_{a} \right)^{2}=\frac{2\pi\mu_{0}\,\omega^{2}R^{6}\lambda_{a}^{2}}{27\,c}.
\end{equation}

\subsection[Exact Shell Power at all kR]{Exact Shell Power at all $kR$}

The dipole result is the long-wavelength limit of one exact multipole. With $\Lambda_{a,\omega}=\lambda_{a}\cos\theta\,\delta(r-R)$ and $\cos\theta=\sqrt{4\pi/3}\,Y_{10}$, the moments in \eqref{eq:Cmultipole} are $q_{\ell m,\omega}=\int j_{\ell}(kr')\Ylm^{*}\Lambda_{a,\omega}\,d^{3}r'$. The radial delta function fixes $r'=R$, and $\int\Ylm^{*}\cos\theta\,d\Omega=\sqrt{4\pi/3}\,\delta_{\ell1}\delta_{m0}$ selects one term.
\begin{equation}
q_{\ell m,\omega}=\lambda_{a}R^{2}j_{\ell}(kR)\sqrt{\tfrac{4\pi}{3}}\,\delta_{\ell1}\delta_{m0}, \qquad q_{10,\omega}=\lambda_{a}R^{2}j_{1}(kR)\sqrt{\tfrac{4\pi}{3}}.
\end{equation}
Only the $\ell=1,m=0$ scalar multipole remains. It radiates except at the ideal-shell zeros of $j_{1}(kR)$. Substitution into \eqref{eq:Pmaster} gives
\begin{equation}
\langle P\rangle_{\mathrm{shell}}=\frac{c\mu_{0}}{2}\,|q_{10,\omega}|^{2} =\frac{c\mu_{0}}{2}\cdot\frac{4\pi}{3}\lambda_{a}^{2}R^{4}j_{1}^{2}(kR) =\frac{2\pi c\mu_{0}}{3}\,\lambda_{a}^{2}R^{4}\,j_{1}^{2}(kR).
\end{equation}
This expression is exact at all $kR$ because the shell keeps the same $\ell=1,m=0$ angular form. Since $j_{1}(x)=\sin x/x^{2}-\cos x/x\to x/3$ as $x\to0$, $j_{1}^{2}(kR)\to(kR)^{2}/9$. The exact result then approaches $\frac{2\pi\mu_{0}\omega^{2}R^{6}\lambda_{a}^{2}}{27c}$, with relative corrections of order $O((kR)^{2})$.

\FloatBarrier

\section{Covariant Lagrangian and Field Equations}
\label{app:lagrangian}
Section~\ref{sec:eed} states that the unconstrained potential equations \eqref{eq:potwave} follow from the Lagrangian density \eqref{eq:lagrangian}. This appendix derives that result, takes the four-divergence to recover \eqref{eq:scalarbox}, and relates the gauge variation to charge conservation. The metric signature is $(+,-,-,-)$, with $x^{\mu}=(ct,\vr)$, $\partial_{\mu}=(c^{-1}\partial_{t},\grad)$, $A^{\mu}=(\Phi/c,\vA)$, $A_{\mu}=(\Phi/c,-\vA)$, $J^{\mu}=(c\rho,\vJ)$, $F_{\mu\nu}=\partial_{\mu}A_{\nu}-\partial_{\nu}A_{\mu}$, and $C=\partial_{\mu}A^{\mu}$. The action is $S=c^{-1}\!\int\mathcal{L}\,d^{4}x$ with $d^{4}x=c\,dt\,d^{3}r$.

\subsection{Euler-Lagrange Equations}

For
\begin{equation}
\mathcal{L}_{\mathrm{EED}}=-\frac{1}{4\mu_{0}}F_{\mu\nu}F^{\mu\nu}-\frac{1}{2\mu_{0}}\,C^{2}-J^{\mu}A_{\mu},
\end{equation}
the momenta conjugate to $A_{\mu}$ are
\begin{equation}
\frac{\partial\mathcal{L}_{\mathrm{EED}}}{\partial(\partial_{\nu}A_{\mu})}
=-\frac{1}{\mu_{0}}F^{\nu\mu}-\frac{1}{\mu_{0}}\,\eta^{\nu\mu}C.
\end{equation}
The first term is the standard Maxwell term. The second follows from $C=\eta^{\alpha\beta}\partial_{\alpha}A_{\beta}$, which gives $\partial C/\partial(\partial_{\nu}A_{\mu})=\eta^{\nu\mu}$. Since $\partial\mathcal{L}_{\mathrm{EED}}/\partial A_{\mu}=-J^{\mu}$, the Euler-Lagrange equation $\partial_{\nu}[\partial\mathcal{L}/\partial(\partial_{\nu}A_{\mu})]-\partial\mathcal{L}/\partial A_{\mu}=0$ gives
\begin{equation}
\partial_{\nu}F^{\nu\mu}+\partial^{\mu}C=\mu_{0}J^{\mu}.
\label{eq:covfield}
\end{equation}
Using $\partial_{\nu}F^{\nu\mu}=\partial_{\nu}(\partial^{\nu}A^{\mu}-\partial^{\mu}A^{\nu})=\partial_{\nu}\partial^{\nu}A^{\mu}-\partial^{\mu}C$, the two $\partial^{\mu}C$ terms cancel. Equation~\eqref{eq:covfield} reduces to
\begin{equation}
\partial_{\nu}\partial^{\nu}A^{\mu}=\mu_{0}J^{\mu}.
\label{eq:boxA}
\end{equation}
\\
For the signature $(+,-,-,-)$, $\partial_{\nu}\partial^{\nu}=c^{-2}\partial_{t}^{2}-\nabla^{2}$, so \eqref{eq:boxA} is \eqref{eq:potwave}. The $\mu=0$ component gives $(\nabla^{2}-c^{-2}\partial_{t}^{2})\Phi=-\rho/\eps$, using $J^{0}=c\rho$ and $\mu_{0}c^{2}=1/\eps$. The spatial components give $(\nabla^{2}-c^{-2}\partial_{t}^{2})\vA=-\mu_{0}\vJ$. No gauge condition was imposed. Once boundary or initial data are specified, the $C^{2}$ term removes the algebraic gauge degeneracy of the kinetic operator and determines the longitudinal part of $A^{\mu}$ dynamically.

\subsection{Scalar Equation}

Taking the four-divergence of \eqref{eq:boxA} and using the commutativity of partial derivatives gives
\begin{equation}
\partial_{\nu}\partial^{\nu}(\partial_{\mu}A^{\mu})=\mu_{0}\,\partial_{\mu}J^{\mu}
\qquad\Longrightarrow\qquad
\partial_{\nu}\partial^{\nu}C=\mu_{0}\Lambda.
\end{equation}
\\
Here $\Lambda=\partial_{\mu}J^{\mu}=\partial_{t}\rho+\dive\vJ$, so this is \eqref{eq:scalarbox}. The same result follows by contracting \eqref{eq:covfield} with $\partial_{\mu}$ and using $\partial_{\mu}\partial_{\nu}F^{\nu\mu}\equiv0$, which follows from the antisymmetry of $F$.

\subsection{Gauge Variation and Charge Conservation}

Under $A_{\mu}\to A_{\mu}+\partial_{\mu}\chi$, the field strength $F_{\mu\nu}$ is unchanged, while $C\to C+\partial_{\nu}\partial^{\nu}\chi$. The interaction term changes as $-J^{\mu}A_{\mu}\to-J^{\mu}A_{\mu}-J^{\mu}\partial_{\mu}\chi$. Integrating the added term by parts and dropping the boundary term for localized sources gives
\begin{equation}
\delta S_{\mathrm{int}}=-\frac{1}{c}\int J^{\mu}\partial_{\mu}\chi\,d^{4}x
=\frac{1}{c}\int(\partial_{\mu}J^{\mu})\,\chi\,d^{4}x=\frac{1}{c}\int \Lambda\,\chi\,d^{4}x.
\end{equation}
For compactly supported $\chi$, or sufficient falloff, the source coupling is invariant under arbitrary admissible gauge functions if and only if $\Lambda=0$. Local charge conservation and gauge invariance of the interaction are therefore the same condition. The quadratic term is not invariant under an arbitrary gauge transformation. Its finite change is $\Delta(C^{2})=2C\,\partial_{\nu}\partial^{\nu}\chi+(\partial_{\nu}\partial^{\nu}\chi)^{2}$. It is invariant for every field configuration only under the residual class $\partial_{\nu}\partial^{\nu}\chi=0$. At fixed $C$, the change also vanishes for $\partial_{\nu}\partial^{\nu}\chi=-2C$, but this field-dependent condition is not a symmetry. Three points are useful. The field-strength term is invariant for every $\chi$. The $C^{2}$ term is invariant under the field-independent residual class $\partial_{\nu}\partial^{\nu}\chi=0$. The source coupling is invariant under arbitrary admissible gauge functions only when $\Lambda=0$. The complete source-coupled action is invariant under a residual transformation only when the interaction term is also unchanged. For $\Lambda\neq0$, even a residual transformation generally changes the action by $c^{-1}\!\int\Lambda\chi\,d^{4}x$. When $\Lambda=0$, the retarded source-generated solution with zero scalar initial data has $C\equiv0$, and no observable depends on it. This is the covariantly gauge-fixed Feynman-gauge setting, where the quadratic term follows after eliminating a Nakanishi-Lautrup field~\cite{Nakanishi1966,Lautrup1967}. Homogeneous scalar data remain mathematically possible. When $\Lambda\neq0$, the transformation is no longer a gauge redundancy of the same source-coupled problem. EED then treats the generated $C$ field as physical, and it radiates when the anomaly has nonzero on-shell moments.

\end{appendices}

\FloatBarrier

\section*{Statements and Declarations}
\noindent\textbf{Funding}\quad No funding was received for this study.

\medskip
\noindent\textbf{Competing Interests}\quad The author declares no competing interests.

\medskip
\noindent\textbf{Ethics Approval and Consent to Participate}\quad Not applicable.

\medskip
\noindent\textbf{Consent for Publication}\quad Not applicable.

\medskip
\noindent\textbf{Data Availability}\quad No datasets were generated or analyzed. The figures are schematic or are direct numerical plots of the closed-form expressions given in the manuscript.

\medskip
\noindent\textbf{Materials Availability}\quad Not applicable.

\medskip
\noindent\textbf{Code Availability}\quad A Python script was used only to prepare the figures. The definitions, derivations, inputs, and conclusions are fully stated in the manuscript and do not depend on access to the script.

\medskip
\noindent\textbf{Author Contributions}\quad N.R. conceived the study, completed the analytical derivations, prepared the figures, and wrote the manuscript.

\FloatBarrier


\begin{thebibliography}{99}
\bibitem{Modanese2017} G. Modanese, Mod. Phys. Lett. B \textbf{31}, 1750052 (2017). \url{https://doi.org/10.1142/S021798491750052X}
\bibitem{Ohmura1956} T. Ohmura, Prog. Theor. Phys. \textbf{16}, 684 (1956). \url{https://doi.org/10.1143/PTP.16.684}. Erratum: Prog. Theor. Phys. \textbf{17}, 131 (1957). \url{https://doi.org/10.1143/PTP.17.131}
\bibitem{AharonovBohm1963} Y. Aharonov, D. Bohm, Phys. Rev. \textbf{130}, 1625 (1963). \url{https://doi.org/10.1103/PhysRev.130.1625}
\bibitem{vanVlaenderen2001} K.J. van Vlaenderen, A. Waser, Hadronic J. \textbf{24}, 609 (2001)
\bibitem{HivelyGiakos2012} L.M. Hively, G.C. Giakos, Int. J. Signal Imaging Syst. Eng. \textbf{5}, 3 (2012). \url{https://doi.org/10.1504/IJSISE.2012.046745}
\bibitem{Okun1989} L.B. Okun, Sov. Phys. Usp. \textbf{32}, 543 (1989). \url{https://doi.org/10.1070/PU1989v032n06ABEH002727}
\bibitem{Borexino2015} M. Agostini et al. (Borexino Collaboration), Phys. Rev. Lett. \textbf{115}, 231802 (2015). \url{https://doi.org/10.1103/PhysRevLett.115.231802}
\bibitem{Majorana2024} I.J. Arnquist et al. (Majorana Collaboration), Nat. Phys. \textbf{20}, 1078 (2024). \url{https://doi.org/10.1038/s41567-024-02437-9}
\bibitem{ReedHively2020} D. Reed, L.M. Hively, Symmetry \textbf{12}, 2110 (2020). \url{https://doi.org/10.3390/sym12122110}
\bibitem{MinottiModanese2021} F. Minotti, G. Modanese, Quantum Rep. \textbf{3}, 703 (2021). \url{https://doi.org/10.3390/quantum3040044}
\bibitem{MinottiModanese2021b} F. Minotti, G. Modanese, Symmetry \textbf{13}, 691 (2021). \url{https://doi.org/10.3390/sym13040691}
\bibitem{MinottiModanese2022} F. Minotti, G. Modanese, Quantum Rep. \textbf{4}, 277 (2022). \url{https://doi.org/10.3390/quantum4030020}
\bibitem{Jackson} J.D. Jackson, Classical Electrodynamics, 3rd edn. (Wiley, New York, 1999)
\bibitem{Griffiths} D.J. Griffiths, Introduction to Electrodynamics, 4th edn. (Pearson, Boston, 2013)
\bibitem{Fermi1932} E. Fermi, Rev. Mod. Phys. \textbf{4}, 87 (1932). \url{https://doi.org/10.1103/RevModPhys.4.87}
\bibitem{Stueckelberg1938} E.C.G. Stueckelberg, Helv. Phys. Acta \textbf{11}, 225 (1938). \url{https://doi.org/10.5169/seals-110852}
\bibitem{Nakanishi1966} N. Nakanishi, Prog. Theor. Phys. \textbf{35}, 1111 (1966). \url{https://doi.org/10.1143/PTP.35.1111}
\bibitem{Lautrup1967} B. Lautrup, Kgl. Danske Videnskab. Selskab, Mat.-Fys. Medd. \textbf{35}, No.~11, 1 (1967)
\bibitem{JimenezMaroto2009} J. Beltr\'an Jim\'enez, A.L. Maroto, J. Cosmol. Astropart. Phys. \textbf{2009}(03), 016 (2009). \url{https://doi.org/10.1088/1475-7516/2009/03/016}
\bibitem{JimenezMaroto2011} J. Beltr\'an Jim\'enez, A.L. Maroto, Phys. Rev. D \textbf{83}, 023514 (2011). \url{https://doi.org/10.1103/PhysRevD.83.023514}
\bibitem{MinottiModanese2023} F. Minotti, G. Modanese, Eur. Phys. J. C \textbf{83}, 1086 (2023). \url{https://doi.org/10.1140/epjc/s10052-023-12274-4}
\bibitem{MinottiModanese2024} F. Minotti, G. Modanese, J. Phys. Commun. \textbf{8}, 055003 (2024). \url{https://doi.org/10.1088/2399-6528/ad4e98}
\bibitem{MinottiModanese2026} F. Minotti, G. Modanese, Int. J. Mod. Phys. A \textbf{41}(1), 2650024 (2026). \url{https://doi.org/10.1142/S0217751X26500247}
\bibitem{MinottiModanese2026b} F. Minotti, G. Modanese, Eur. Phys. J. Plus \textbf{141}, 788 (2026).
\url{https://doi.org/10.1140/epjp/s13360-026-08031-7}
\bibitem{MinottiModanese2023b} F. Minotti, G. Modanese, Symmetry \textbf{15}, 1119 (2023). \url{https://doi.org/10.3390/sym15051119}
\bibitem{Ignatiev1979} A.Yu. Ignatiev, V.A. Kuzmin, M.E. Shaposhnikov, Phys. Lett. B \textbf{84}, 315 (1979). \url{https://doi.org/10.1016/0370-2693(79)90048-0}
\bibitem{Dubovsky2000} S.L. Dubovsky, V.A. Rubakov, P.G. Tinyakov, J. High Energy Phys. \textbf{2000}(08), 041 (2000). \url{https://doi.org/10.1088/1126-6708/2000/08/041}
\bibitem{MinottiModanese2025} F. Minotti, G. Modanese, Mathematics \textbf{13}, 892 (2025). \url{https://doi.org/10.3390/math13050892}
\bibitem{Wei2016} Y. Wei, Phys. Rev. E \textbf{93}, 066103 (2016). \url{https://doi.org/10.1103/PhysRevE.93.066103}
\bibitem{Modanese2018} G. Modanese, Mathematics \textbf{6}, 155 (2018). \url{https://doi.org/10.3390/math6090155}
\bibitem{Modanese2017b} G. Modanese, Results Phys. \textbf{7}, 480 (2017). \url{https://doi.org/10.1016/j.rinp.2017.01.009}
\bibitem{Gupta1950} S.N. Gupta, Proc. Phys. Soc. A \textbf{63}, 681 (1950). \url{https://doi.org/10.1088/0370-1298/63/7/301}
\bibitem{Bleuler1950} K. Bleuler, Helv. Phys. Acta \textbf{23}, 567 (1950). \url{https://doi.org/10.5169/seals-112124}
\bibitem{Woodside1999} D.A. Woodside, J. Math. Phys. \textbf{40}, 4911 (1999). \url{https://doi.org/10.1063/1.533007}
\bibitem{Woodside2000} D.A. Woodside, J. Math. Phys. \textbf{41}, 4622 (2000). \url{https://doi.org/10.1063/1.533368}
\bibitem{Woodside2009} D.A. Woodside, Am. J. Phys. \textbf{77}, 438 (2009). \url{https://doi.org/10.1119/1.3076300}

\bibitem{KellerHively2019} O. Keller, L.M. Hively, J. Phys. Commun. \textbf{3}, 115002 (2019). \url{https://doi.org/10.1088/2399-6528/ab5189}
\bibitem{HivelyLand2021} L.M. Hively, M. Land, J. Phys.: Conf. Ser. \textbf{1956}, 012011 (2021). \url{https://doi.org/10.1088/1742-6596/1956/1/012011}
\end{thebibliography}
\end{document}